\documentclass{article}[11pt,letter]

\usepackage[T1]{fontenc}
\usepackage[english]{babel}
\usepackage[cmex10]{amsmath}
\usepackage[utf8]{inputenc}
\usepackage{amssymb}
\usepackage{amsthm}
\usepackage{mathtools}
\usepackage{xspace}
\usepackage[symbol*]{footmisc}
\usepackage[hyperfootnotes=false,breaklinks,bookmarks=false]{hyperref}
\usepackage{algorithmic}
\usepackage{algorithm}
\usepackage{color}
\usepackage{authblk}
\usepackage{fullpage}
\usepackage{cite}
\usepackage{todonotes}

\usepackage{microtype}%if unwanted, comment out or use option "draft"

\newtheorem{definition}{Definition}[section]
\newtheorem{lemma}[definition]{Lemma}
\newtheorem{theorem}[definition]{Theorem}

\newtheorem{observation}[definition]{Observation}

\newcommand{\bigo}{\mathcal{O}}
\newcommand{\ie}{i.e.\xspace}

\newcommand{\alg}{\mathcal{A}}
\newcommand{\dalg}{\mathcal{D}}
\newcommand{\head}{h}
\newcommand{\hash}{\Phi}
\newcommand{\divides}{\bigm|}

\newcommand{\poly}{\text{poly}}
\newcommand{\bhash}{\hash}
\newcommand{\midpalin}[1][n]{{\bf MID-PALIN}$_{\Sigma}[#1]$\xspace}
\newcommand{\palin}[1][n]{{\bf PALIN}$_{\Sigma}[#1]$\xspace}
\newcommand{\aerr}{\ensuremath{E}}
\newcommand{\etal}{et~al.}
\newcommand{\level}{\lambda}
\DeclareMathOperator{\per}{per}
\DeclareMathOperator{\lcm}{lcm}
\newcommand{\kpers}{\ensuremath{p}}
\newenvironment{mydesc}[1]{%
  \begin{list}{}{%
    \settowidth{\labelwidth}{#1}
    \setlength{\labelsep}{0.5cm}
    \setlength{\leftmargin}{\labelwidth}
    \addtolength{\leftmargin}{\labelsep}
    \setlength{\rightmargin}{0pt}
    \setlength{\parsep}{0.5ex plus0.2ex minus0.1ex}
    \setlength{\itemsep}{0ex plus0.2ex}
    }
  }
{\end{list}}

\usepackage{thmtools}
\usepackage{thm-restate}

\usepackage{booktabs}
\usepackage{multirow}

\begin{document}

\title{Tight tradeoffs for approximating palindromes in streams\footnote{up to a logarithmic factor.}}

\author[1]{Paweł Gawrychowski}
\author[2]{Przemysław Uznański}
\affil[1]{Institute of Informatics, University of Warsaw, Poland}
\affil[2]{Department of Computer Science, ETH Z\"{u}rich, Switzerland
}

\date{}

\maketitle

\begin{abstract}
We consider computing the longest palindrome in a text of length $n$ in the streaming model, where
the characters arrive one-by-one, and we do not have random access to the input. While 
computing the answer exactly using sublinear memory is not possible in such a setting, one can still hope for a good approximation guarantee.

We focus on the two most natural variants, where we aim for either additive or multiplicative approximation
of the length of the longest palindrome. We first show that there is no point in considering
Las Vegas algorithms in such a setting, as they cannot achieve sublinear space complexity.
For Monte Carlo algorithms, we provide a lower bound of $\Omega(\frac{n}{\aerr})$ bits for
approximating the answer with additive error $\aerr$, and  $\Omega(\frac{\log n}{\log(1+\varepsilon)})$ bits for 
approximating the answer with multiplicative error $(1+\varepsilon)$ for the binary alphabet.
Then, we construct a generic Monte Carlo algorithm, which by 
choosing the parameters appropriately achieves space complexity matching up to a logarithmic factor for
both variants.
This substantially improves the previous results by Berenbrink~\etal~(STACS~2014) and essentially settles
the space complexity.
\end{abstract}

\section{Introduction}
A recent trend in algorithms on strings is to develop efficient algorithms in the streaming model, where
characters arrive one-by-one, and we do not have random access to the input. The main goal is to minimize
the space complexity, \ie, avoid storing the already seen prefix of the text explicitly. One usually allows
randomization and requires that the answer should be correct with high probability.
We consider computing the longest palindrome in this model, where a palindrome is a fragment which
reads the same in both directions. This is one of the basic questions concerning regularities in texts
and it has been extensively studied in the classical non-streaming setting, see~\cite{Apostolico,GalilSeiferas,KMP,Manacher}
and the references therein. The notion of palindromes, but with a slightly different meaning, is
very important in
computational biology, where one considers strings over $\{A,T,C,G\}$ and a palindrome is a sequence
equal to its reverse complement (a reverse complement reverses the sequences and interchanges
$A$ with $T$ and $C$ with $G$); see~\cite{GawrychowskiSTACS} and the references therein for
a discussion of their algorithmic aspects. Our results generalize to biological palindromes
in a straightforward manner.

Computing the longest palindrome in the streaming model was recently considered by Berenbrink \etal~\cite{Berenbrink},
who developed tradeoffs between the bound on the error and the space complexity for additive and multiplicative variants of the problem, that is,
for approximating the length of the longest palindrome with either additive or a multiplicative error.
Their algorithms were Monte Carlo, \ie, returned the correct answer with high probability. They also proved
that any Las Vegas algorithm achieving additive error $\aerr$ must necessarily use $\Omega(\frac{n}{\aerr}\log|\Sigma|)$ bits of memory,
which matches the space complexity of their solution up to a logarithmic factor in the $\aerr\in [1,\sqrt{n}]$ range,
but leaves at least two questions. Firstly, does the lower bound still hold for Monte Carlo algorithms? Secondly, 
what is the best possible space complexity when $\aerr\in (\sqrt{n},n]$ in the additive variant,
and what about the multiplicative version? We answer all these questions.
%, essentially settling the space
%complexity of approximating the length of the longest palindrome in the streaming model.

\paragraph{Related work.} The most basic problem in algorithms on strings is pattern matching, where we want to detect an occurrence of a a pattern in a given text.
It is somewhat surprising that one can actually solve it using polylogarithmic space in the streaming model,
as proved by Porat and Porat~\cite{PoratStreaming}. A simpler solution was later given by Erg{\"u}n \etal~\cite{ErgunPeriodicity}, and Breslauer
and Galil~\cite{Breslauer}. Similar questions studied in such setting include multiple-pattern
matching~\cite{DictionaryStream}, approximate pattern matching~\cite{KMismatch}, and
parametrized pattern matching~\cite{ParametrizedStreaming}.

Pattern matching is also very closely related to detecting periodicities, and in fact Erg{\"u}n \etal~\cite{ErgunPeriodicity} also developed
an efficient algorithm for computing the smallest period, where $p$ is a period of $T[1..n]$ if $T[i]=T[i+p]$ for all $i=1,2,\ldots,n-p$. Also palindromes are closely connected to periodicities. Informally, two long
palindromes occurring close to each other imply a periodicity of the underlying fragment of the text
(to the best of our knowledge, this has been first explicitly stated by Apostolico \etal~\cite{Apostolico}).
Similar insights have been used to partition the text into the smallest number of palindromes~\cite{FiciGKK14,ISIBT14}
and recognizing the so-called $\text{Pal}^{k}$ language~\cite{KosolobovRS15}.
At a very high level, the idea there is to consider longer and longer prefixes of the text and maintain a succinct description
of all palindromic suffixes of the current prefix. Naturally, our algorithm is based on the same high-level idea, but there
are multiple non-trivial technical difficulties stemming from the fact that we cannot provide random access to the already
seen part of the text, so we can only approximate such information.

\paragraph{Model.} We work in the streaming model and consider additive and multiplicative variant of the problem.
The model works as follows: we are first given the length of the text $n$ and the bound on the desired error
$\aerr$ (in the additive variant) or $\varepsilon$ (in the multiplicative variant), then the characters
$T[1],T[2],\ldots,T[n]\in\Sigma$ arrive one-by-one. In the $h$-th step we receive $T[\head]$ and we are required to output a number $\ell$, such that the length of the longest palindrome in
$T[1..\head]$ is either between $\ell$ and $\ell+\aerr$ (in the additive variant) or between $\ell$ and $(1+\varepsilon)\cdot \ell$
(in the multiplicative variant). 
We have $s(n)$ bits of memory available, where we can store an arbitrary data.
It is important to remember that the procedure operates in steps corresponding
to the characters and we cannot retrieve an already seen character unless it has been stored in memory.

Now we are interested in the possible tradeoffs between $s(n)$ and the bound on the error.
%Intuitively, the larger the allowed error is, the less space we need, but what is the exact relation?
We consider Las Vegas and Monte Carlo algorithms.
A Las Vegas algorithm always returns a correct answer, but its memory usage $s(n)$ is a random variable.
A Monte Carlo algorithm returns a correct answer with high probability, and its memory usage $s(n)$
does not depend on the random choices, where high probability means $1-\frac{1}{n^{c}}$, for arbitrarily large constant $c$.
%Formally, $s(n)=\bigo(f(n))$ with
%high probability if for any $c$ there exists $\gamma(c)$ such that $s(n)\leq \gamma(c)f(n)$.
%Thus, all asymptotic notation hides dependency on $c$, unless stated otherwise.

We assume that the memory consists of words of size $\Omega(\log\max\{n,|\Sigma|\})$ and
basic operations take $\bigo(1)$ time on such words. Bounds on the space are expressed in such words unless stated otherwise.

\paragraph{Previous work.} The longest palindrome can be found in $\bigo(n)$ time and space (cf. Manacher \cite{Manacher}).
Berenbrink \etal~\cite{Berenbrink} constructed a streaming algorithm
achieving additive error $\aerr$ using $\bigo(\frac{n}{\aerr})$ space and $\bigo(\frac{n^{1.5}}{\aerr})$ total time
for any $\aerr \in [1,\sqrt{n}]$,
and a streaming algorithm guaranteeing multiplicative error $(1+\varepsilon)$ using $\bigo(\frac{\log n}{\varepsilon \log(1+\varepsilon)})$
space and $\bigo(\frac{n \log n}{\varepsilon \log(1+\varepsilon)})$ total time for any $\varepsilon \in (0,1]$,
both Monte Carlo.
They also proved that any Las Vegas algorithm with additive error $\aerr$ must necessarily use
$\Omega(\frac{n}{\aerr}\log|\Sigma|)$ bits of space.

\paragraph{Our results.} We significantly improve on the previous results as follows and essentially settle the space complexity of the
problem in both variants (see Table~\ref{table:results} for summary).

Firstly, we prove that any Las Vegas algorithm approximating (in either variant) the length of the longest palindrome inside a text of length $n$
over an alphabet $\Sigma$ must necessarily use $\Omega(n\log|\Sigma|)$ bits of memory (see Theorem~\ref{th:lv}).
Hence Las Vegas randomization is simply not the right model for this particular problem.
Then we move to Monte Carlo algorithms, and prove the following lower bounds on their space complexity:
\begin{itemize}
\item $\Omega(\frac{n}{\aerr} \log \min \{|\Sigma|,\frac{n}{\aerr}\} )$ bits to achieve additive error $\aerr$ with high probability if $\aerr \in [1,0.49 n]$ (see Theorem~\ref{th:montecarlo_additive_lowerbound}),\footnote{This can be strengthened to $0.5n - \omega(\sqrt{n})$, but for the sake of clarity we prefer to state a weaker bound, here and in subsequent similar places.}
\item $\Omega(\frac{\log n}{\log(1+\varepsilon)}\log \min \{|\Sigma|,\frac{\log n}{\log(1+\varepsilon)}\})$ bits to achieve multiplicative error $(1+\varepsilon)$ with high probability if $\varepsilon \in [n^{-0.98}, n^{0.49}]$\  (see Theorem~\ref{th:montecarlo_multiplicative_lowerbound}).\footnote{Here $-0.98$ can be replaced by any constant larger than $-1$, and $n^{0.49}$ can be strengthened to $o(\sqrt{n})$.}
\end{itemize}

\newcommand{\tablebgcolor}{}
\newcommand{\tablebgcolordark}{}
\newcommand{\cellbgcolor}{}
\newcommand{\cellbgcolordark}{}

\begin{table}[t!]

\begin{tabular}{p{0.44\columnwidth}p{0.51\columnwidth}}
	\toprule		
		\multicolumn{2}{c}{Las Vegas approximation}\\
		\midrule
		$\Omega(\frac{n}{\aerr}\log|\Sigma|)$ & $\Omega(n\log|\Sigma|)$ \\
		\midrule
		\multicolumn{2}{c}{Monte Carlo additive approximation}\\
		\midrule
		$\bigo(\frac{n}{\aerr})$ space, $\bigo(\frac{n^{1.5}}{\aerr})$ time, $\aerr \in [1,\sqrt{n}]$ & $\bigo(\frac{n}{\aerr})$ space, $\bigo(n \log n)$ time, $\aerr \in [1,n]$\\
	
		\multicolumn{1}{c}{--} & $\Omega(\frac{n}{\aerr} \log \min\{ |\Sigma| \}, \frac{n}{\aerr} \})$ bits, $\aerr \in [1,0.49n]$ \\
		\midrule
		\multicolumn{2}{c}{Monte Carlo multiplicative approximation}\\
		\midrule
		$\bigo(\frac{\log n}{\varepsilon\log(1+\varepsilon)})$ space, $\bigo(\frac{n\log n}{\varepsilon \log(1+\varepsilon)})$ time, $\varepsilon \in (0,1]$ & $\bigo(\frac{\log(n\varepsilon)}{\log(1+\varepsilon)})$ space, $\bigo(n \log n)$ time, $\varepsilon \in [\frac2n,n]$\\

		\multicolumn{1}{c}{--} & $\Omega(\frac{\log n}{\log(1+\varepsilon)}\log \min \{|\Sigma|,\frac{\log n}{\log(1+\varepsilon)}\})$ bits, $\varepsilon \in [n^{-0.98}, n^{0.49}]$\\
		\bottomrule
	\end{tabular}
	\caption{
    A comparison of previous (on the left, c.f. \cite{Berenbrink}) and our (on the right) results. Lower bounds are in bits, and upper bounds in words consisting of $\Omega(\log \max\{ n,|\Sigma|\})$ bits.
\label{table:results}
}

\end{table}

Secondly, we construct a generic 
Monte Carlo approximation algorithm, which by adjusting the parameters appropriately matches our lower bounds up to a logarithmic multiplicative 
factor.
In more detail, our algorithm uses $\bigo(\frac{n}{\aerr})$ words of space for any $\aerr \in [1,n]$ in the additive variant (see Theorem~\ref{additive_scheme} and
Theorem~\ref{memory_efficient}) and $\bigo(\frac{\log (n\varepsilon)}{\log(1+\varepsilon)})$ words of space for any $\varepsilon \in [\frac2n,n]$
in the multiplicative variant (see Theorem~\ref{multiplicative_scheme} and Theorem~\ref{memory_efficient}).\footnote{For 
small $\varepsilon$ this is $\bigo(\frac{\log(n\varepsilon)}{\varepsilon})$, and for large $\varepsilon$ becomes $\bigo(\frac{\log n}{\log(1+\varepsilon)})$.
Note that this does not contradict the lower bound, because $\log(n\varepsilon)=\Theta(\log n)$ for $\varepsilon\in[n^{-0.98},n^{0.49}]$.}
This essentially settles the space complexity of the problem, as it can be seen that our lower and upper bounds differ by at most a logarithmic factor
for any $\aerr \in [1,0.49n]$ and $\varepsilon \in [n^{-0.98},n^{0.49}]$.
The time complexity of our algorithm is always $\bigo(n\log n)$ (see Theorem~\ref{thm:time_efficient}).

\paragraph{Overview of the methods.} As usual in the streaming model, we apply Karp-Rabin fingerprints. We store
such fingerprints for some carefully chosen prefixes of the already seen part of the text. Informally, these chosen
prefixes become more and more sparse as we move closer to the beginning, with the details depending on the
variant. We call such fingerprints of prefixes landmarks, and formalize this notion in Section~\ref{section:preliminaries}. Then, in Section~\ref{section:basic}, we present
a generic algorithm. The idea is that for every possible palindrome center we create a separate process
which maintains the corresponding palindromic radius (or, more precisely, its approximation).
By adjusting the parameters of the generic algorithm we are able to guarantee good bound on the error
in both variants.
For $\aerr \in \Omega(\frac{n}{\log n})$ or $\varepsilon \in \Omega(1)$ there are only few landmarks and such generic algorithm
is already efficient enough when implemented naively. To implement it efficiently for smaller $\aerr$
or $\varepsilon$, we need to avoid running processes which have already found
a mismatch, but maintaining such a list explicitly might take too much space. However, multiple sufficiently
long palindromes appearing close to each other imply periodicity of the corresponding fragment of the text, which
can be exploited to concisely describe the whole situation.
This insight (dating back to Apostolico \etal~\cite{Apostolico})
allows us to approximate the information about all active processes in logarithmic space as explained
in Section~\ref{section:time}. Then in Section~\ref{section:time2} we use it to avoid running all active processes after
reading every character.
Finally, in Section~\ref{section:lowerbounds} we apply the Yao's minimax principle to derive the
lower bounds. For Las Vegas algorithms, this is straightforward, but requires more work for Monte Carlo algorithms.

%Because of limited space, many proofs are deferred to the appendix. For convenience of the reader we also provide the full version~\cite{naszarxiv}.

\paragraph{Comparison with previous work.}
The additive approximation algorithm proposed by Berenbrink \etal~\cite{Berenbrink} uses a flat structure of $\sqrt{n}$ fingerprints (called checkpoints).  The most recently seen $\sqrt{n}$ characters are stored explicitly,
and the information of all palindromes with larger radius is compressed using the periodicity lemma.
For multiplicative approximation, a sparse structure of checkpoints is used.
Our technical contribution is of several flavors. Firstly, we use a single generic construction 
for both variants of the problem. Secondly, in all variants we use a hierarchical structure of fingerprints,
with the fingerprints becoming more and more sparse as we move closer to the beginning.
This in particular allows us to avoid storing a long suffix explicitly in the
additive version. Thirdly, also the periodicity compression is applied in a hierarchical manner:
we maintain a partition of the text into segments with lengths exponentially increasing
with the distance from the most recently seen character and compress each such segment
separately. In previous work, a rigid partition into segments of length $\sqrt{n}$ was used.
Storing such rigid partition requires $\Omega(\sqrt{n})$ space, which might be too much
when the allowed error is large. Working with segments of exponentially increasing lengths
allows us to decrease this additional memory usage to only $\bigo(\log n)$, but
also makes the details more involved.
%Finally, our structure of landmarks (checkpoints) is less rigid than in \cite{Berenbrink}, which allows us to reduce the time cost in amortized analysis.\NOTE{w multiplicative maja landmarki ustawione na sztywno kiedy wygasaja, co chyba jest jednym z powodow dla ktorego jest powolnie i duzo miejsca}

%For lowerbounds, we start by showing that Las Vegas randomization cannot provide any gain in memory usage over non-streaming solution, regardless of approximation model. Since that lowerbound exceeds upperbound presented in \cite{Berenbrink}, any claims of tightness of Monte Carlo algorithms made there, based on matching lowerbound, were premature. We then proceed to provide a linear in $n$  lowerbound for the memory usage of any Monte Carlo streaming algorithms for finding exact palindromic radius for an explicitly stated beforehand point of input. We then proceed with a series of carefully constructed reductions using padding, showing almost-tightness of upperbounds we provide, for a wide range of parameters.

\section{Preliminaries}
\label{section:preliminaries}

For a word $w\in\Sigma^{*}$, we denote its length by $|w|$, and its $i$-th letter by $w[i]$ for any $i=1,2,\ldots,|w|$. Similarly, $w[i..j]$ denotes the fragment starting at the $i$-th
and ending at the $j$-th character, and $w^{R}$ denotes the reversal, \ie, $w[|w|]w[|w|-1]..w[1]$. The period $\per(w)$ of $w$ is the smallest natural number such that
$w[i]=w[i+\per(w)]$ for all $i=1,2,\ldots,|w|-\per(w)$.
The well-known periodicity lemma~\cite{FineWilf} states that if $p$ and $q$ are periods of $w$ and $p+q\leq |w|$, then $\gcd(p,q)$ is also a period
of $w$.
We focus on detecting palindromes of even length (odd palindromes can be detected with the standard trick of duplicating every letter, see~\cite{Apostolico}).
The palindromic radius at $c$ is the largest $R(c)$ such that $T[c..(c+R(c)-1)] = T[(c-R(c))..(c-1)]^R$. $c$ is the center of
a~palindrome $T[(c-R(c))..(c+R(c)-1)]$.

\paragraph{Karp-Rabin fingerprints.} We use the Karp-Rabin fingerprints~\cite{KR} to quickly check equality of long strings. We choose a large prime $p\geq \max(|\Sigma|,\poly(n))$
and draw $x\in \mathbb{Z}_{p}$ uniformly at random.% where $n$ is a sufficiently large constant.
Then define $f_x(w[1..k]) = (w[1] + w[2] x + \ldots + w[k] x^{k-1}) \bmod p$
and define fingerprint $\hash(w)$ of a word $w$ to consist of $|w|$, $f_x(w)$, $f_{x}(w^R)$, $x^{|w|}$, $x^{-|w|}$,
which takes $\bigo(1)$ space if $|w|\leq n$.
The following operations take $\bigo(1)$ time:
\begin{mydesc}{\bf erasing a suffix}
\item[\bf concatenation] given $\hash(w)$ and $\hash(v)$, find $\hash(wv)$,
\item[\bf erasing a prefix] given $\hash(wv)$ and $\hash(w)$, find $\hash(v)$,
\item[\bf erasing a suffix] given $\hash(wv)$ and $\hash(v)$, find $\hash(w)$,
\item[\bf reversal] given $\hash(w)$, find $\hash(w^R)$.
\end{mydesc}

The fingerprints allow us to check if two strings are the same. Formally, we assume that $|\Sigma|\leq\poly(n)$. Then,
to check if $u=v$ we compare $\hash(u)$ and $\hash(v)$.
If $u=v$ then $\hash(u)=\hash(v)$, and if $u\neq v$ while $|u|,|v|\leq n$ then $\hash(u)=\hash(v)$ with probability at most $\frac{n}{\poly(n)}$.
The latter situation is called a false positive. Because the running time of our algorithms will be always polynomial in $n$, and we will
be operating on strings of length at most $n$, by the union bound the probability of a false positive can be made $\frac{1}{n^{c}}$
for any $c$ by choosing exponent in $\poly(n)$ large enough. When analyzing the correctness, we assume no false positives.

\paragraph{Landmarks.} Our algorithms stores some values $\bhash(i)=\hash(T[1..i])$.
The intuition is that after reading $T[\head]$ we calculate
the fingerprint of the currently seen prefix and keep it available for some time.
A landmark is a position $i$ such that $\bhash(i)$ is currently stored.
If additionally $i=2^{\level}\cdot j$ for some $j$, $i$ is a landmark on
level $\level$, and if $j$ is odd then $i$ is a landmark strictly on level
$\level$. $\mathcal{Y}_\level$ is the set of landmarks on level $\level$
and $\mathcal{Y}$ is the set of all landmarks. Observe that knowing
$\bhash(t)$ for all $t\in\mathcal{Y}$ is enough to calculate
$\hash(T[t+1\ ..\ t'])$ for any $t,t' \in \mathcal{Y}$ in $\bigo(1)$ time.
Technically, $\mathcal{Y}$ depends on the current value of $\head$, and
``$t$ is a landmark at $h$'' means that $t\in\mathcal{Y}$ just after reading $T[\head]$.

For each level of landmarks we fix its size $b_{\level}$. After reading $T[\head]$, the
$\level$-th level consists of the $b_\level$ most recently seen positions of the form $2^\level \cdot j$,
that is, $2^\level (\lfloor \frac{\head}{2^\level} \rfloor - i)$ for $i = 0, 1, \ldots, b_\level-1$.
The sizes $b_\level$ are chosen differently depending on the desired approximation guarantee.
In all versions, the last level has number $L\le \log n$, and $b_0=...=b_{L-1}$, while no restriction is put on $b_L$.
%there is $L\leq \log n$ such that $b_0=b_1=\ldots b_{L-1}$ and $L$ is the last
%level with possibly larger value of $b_L$.

% Pawel: poniższa uwaga jest potrzebna?
% we require that either 
% $\level$ is the last level and we have all possible landmarks there, or the leftmost landmark on level % $\level+1$ is strictly
% on the left of all landmarks on level $\level$. 
For such choice of landmarks, after increasing $h$ by one we need to add and remove at
most one landmark per level, which takes $\bigo(\log h)$ time in total. The landmarks
on each level are kept in a random access array with cyclic addressing, thus using 
$\bigo(b_{\level}+1)$ space while allowing accesses and updates in $\bigo(1)$ time.

\section{Space-efficient algorithm}
\label{section:basic}
We start with the basic algorithm. A proper choice of all $b_{\level}$ guarantees small additive or multiplicative error, but the time and
space complexity might be high. Nevertheless, the basic algorithm serves as a good starting point for developing first the space efficient version,
and then finally the time efficient solution.

The idea of the algorithm is that for every possible center $c$ we create a process $P(c)$, which keeps on computing the corresponding
radius $R(c)$. We call a process alive if it has not found $T[c+\Delta]$ such that
$T[c-1-\Delta]\neq T[c+\Delta]$ yet, and dead otherwise. The process starts with $R(c)=0$ and then uses the landmarks to update the value
of $R(c)$ (and also the final answer) whenever possible.
To verify if $R(c) \geq \head-c+1$ we need to check if $T[(2c-\head-1)..h]$ is a palindrome. This requires accessing
$\bhash(\head)$ and $\bhash(2c-\head-2)$ to calculate $\hash(T[(2c-\head-1)..\head])$ and then $\hash((T[(2c-\head-1)..\head])^{R})$, which
are then compared to each other, see Fig.~\ref{fig:radius}.
We can simply maintain the current value of $\bhash(\head)$ but retrieving $\bhash(2c-\head-2)$ is only possible when $2c-\head-2$ is a landmark.
Therefore, the process $P(c)$ can update its $R(c)$ only when $2c-\head-2$ is a landmark, and doing so will be referred to as running $P(c)$ using $2c-\head-2$. If $T[(2c-\head-1)..\head]$
is a palindrome, we say that $P(c)$ succeeds, and otherwise fails.

\begin{figure}[t]
\includegraphics[width=\textwidth]{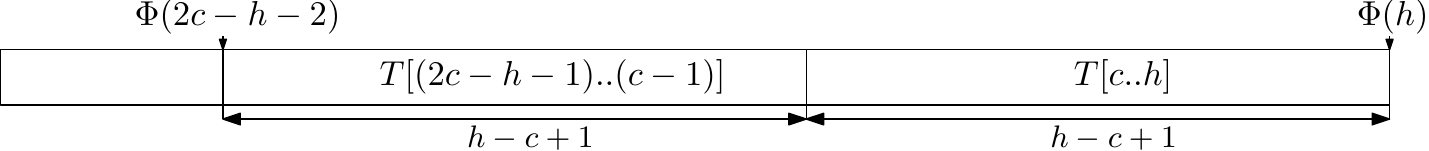}
\caption{Checking if $R(c)\geq \head-c+1$.}
\label{fig:radius}
\end{figure}

We would like to guarantee that running a process is a $\bigo(1)$ time procedure, so we need to quickly check if $2c-\head-2$
is currently a landmark (and if so, access the stored $\bhash(2c-\head-2)$). This can be easily done by iterating through all possible levels, but we want a
faster method. We consider the last $L$-th level separately in $\bigo(1)$ time. For all lower levels, the values $b_{\level}$ are all the same. We compute
the largest power of $2$ dividing $2c-\head-2$, call it $2^{\level}$, then $2c-\head-2$ cannot be a landmark on level larger than $\level$.
On the other hand, if $2c-\head-2$ is a landmark on level $\level'<\level$, then it is also a landmark on level $\level$. Therefore, we only need to consider the
$\level$-th level. This allows us to run any process in $\bigo(1)$ time, and furthermore the state of any $P(c)$
can be fully described just by specifying its center $c$ ($R(c)$ is not stored explicitly, unless mentioned otherwise).
Observe that even if a process is dead, there is no harm in running it again.%, as it will fail. 
%Such situation will be called running a zombie.

The basic algorithm simply runs all processes after reading the next $T[\head]$ using appropriately defined landmarks (depending on variant
and desired error guarantee).
For $\aerr \in \Omega(\frac{n}{\log n})$ or $\varepsilon \in \Omega(1)$, it needs $\bigo(\log n)$ time to process $T[\head]$.
%For smaller $\aerr$ and $\varepsilon$, we will modify it as to run just a few processes.

\begin{restatable}{theorem}{restatablespeff}
\label{th:sp_eff}
The basic algorithm approximates the longest palindrome with additive error $\aerr \in [1,n]$ using $\bigo(\frac{n}{\aerr})$ time and words of memory
to process $T[\head]$, and the longest palindrome with multiplicative error $\varepsilon \in [\frac2n,n]$ using $\bigo(\frac{\log(n\varepsilon)}{\log(1+\varepsilon)})$
time and words of memory to process $T[\head]$.
\end{restatable}

\begin{proof}
The algorithm runs, after reading every $T[\head]$, the process $P(\frac{\head+y}{2})$ for every landmark $y$.
For additive approximation, the landmarks are defined as in Theorem~\ref{additive_scheme} with $L=\log \aerr$.
For multiplicative approximation and $\frac2n \le \varepsilon \le 1$,
the landmarks are defined as in Theorem~\ref{multiplicative_scheme} with $D=\Theta(\frac{1}{\varepsilon})$,
but when $1 \le \varepsilon \le n$ we slightly change the definition by choosing $k = \log (1+\varepsilon)$ and keeping as a landmark,
for every $i=1,2,\ldots,\frac{\log n}{k}$, the last position divisible by $2^{k\cdot i}$.
\end{proof}

An optimized version of the basic algorithm will be referred
to as a scheduling scheme. A scheduling scheme should guarantee that any alive $P(c)$ such that $2c-\head-2$ is a landmark is run, unless we 
can either be sure that it would fail anyway, or there is another $P(c')$ such that $c'<c$ and $2c'-\head-2$ is also a landmark, and we can be sure
that it succeeds (in particular, $P(c')$ is still alive).

\begin{restatable}{theorem}{restatableadditivescheme}
\label{additive_scheme}
Any scheduling scheme with $b_L = \infty$ approximates the longest palindrome with additive error
$2^L$.
\end{restatable}

\begin{proof}
Consider an arbitrary palindrome $T[(c-x)..(c+x-1)]$. We will show that any scheduling scheme with $b_{L}=\infty$ returns at least $x-2^{L}$.
We can assume $x \geq 2^{L}$, then there exists $y\in (x-2^{L},x]$ such that $2^{L}\divides c-y-1$.
Therefore, $c-y-1$ is permanently a landmark on level $L$. $P(c)$ will be alive at $c+y-1$, so by the properties of a scheduling
scheme we will run a process detecting a~palindrome with radius at least $y > x-2^{L}$.
Thus any scheduling scheme with $b_{L}=\infty$ approximates the longest palindrome with additive error $2^{L}$.
\end{proof}

%\vspace{-1.9ex}
\begin{restatable}{theorem}{restatablemultiplicativescheme}
\label{multiplicative_scheme}
Any scheduling scheme with $b_0 = b_{1} = \ldots, b_{\log(n/D)} = D$ for $D\geq 6$ approximates the longest palindrome with multiplicative 
error $1+\bigo(1/D)$.
\end{restatable}

\begin{proof}
Consider an arbitrary palindrome $T[(c-x)..(c+x-1)]$. We will show that any such scheduling scheme returns at least $x/(1+\bigo(1/D))$.
Let $\level$ be the smallest integer such that $(D-1) \cdot 2^\level \ge 2 \cdot x$. We have two cases.

\begin{description}
\item[$\level=0$] After reading $T[c+x-1]$, all $c+x-1,c+x-2,\ldots,c+x-D$ are landmarks on level 0, and $c-x-1$ is one
of them because $2x < D$, so $T[(c-x)..(c+x-1)]$ or a longer palindrome is detected.

\item[$\level>0$] In the interval $[c-x-1,c+x-1]$ there are at most $\left\lfloor \frac{2x}{2^\level} \right\rfloor +1 \le D$  numbers divisible by $2^{\level}$,
thus there exists $y \in (x-2^{\level},x]$ such that $c-y-1$ was a landmark on level $\level$ after reading $T[c+x-1]$. As $P(c)$ is still alive
at $c+x-1$, we will detect a palindrome of radius at least $y$. In other words, we will approximate $x$ with additive error $2^{\level}$.
However, since $\level$ was chosen to be minimal, we have that $2 \cdot x > (D-1) \cdot 2^{\level-1}$, so we can bound the multiplicative error from above by
$\frac{x}{x-2^{\level}}$, which is at most $1+\frac{1}{\Omega(D)}$ for $D\geq 6$.
\end{description}
Therefore, any such scheduling scheme approximates the longest palindrome with multiplicative error $1+\bigo(1/D)$.
\end{proof}

%\vspace{-1.9ex}
%The following technical lemma will be useful later.
\begin{restatable}{lemma}{restatabledelay}
\label{lemma:delay}
If $b_{\level}\geq 12$ for all $\level\leq L$, then for any $\Delta = 2^{\ell}-1$ and for any $c$ there is at least one $h \in [c+5\Delta,c+6\Delta]$ such that 
$2c-\head-2$ is a landmark at $\head$.
\end{restatable}

\begin{proof}
Observe that there exists unique $h \in [c+5\Delta,c+6\Delta]$ such that $2^{\ell} \divides 2c-h-2$. Because $h - (2c-h-2) \le 2+12\cdot \Delta < 12\cdot2^{\ell}$, after reading $T[\head]$ there are two possibilities. If $\ell \le L$ then $2c-h-2$ is among the 12 last seen positions divisible by $2^\ell$. If $\ell>L$ then $2c-h-2$ is definitely a landmark anyway.
\end{proof}

\section{Maintaining the alive processes}
\label{section:time}

To implement a scheduling scheme, we want to know which processes are alive. Maintaining
them explicitly is too space-expensive, though.
Therefore, we will store a compressed approximate representation of all alive processes.
Intuitively, we will group together nearby alive processes using the notion
of a partition scheme described below. The representation will not be exact
in the sense that it might report some dead processes as alive (but then they will
have some additional properties).
Such information is not providing any speedup by itself yet,
but later in Section~\ref{section:time2} we will use it to implement any scheduling scheme efficiently.
In this section, we focus on maintaining the information.% efficiently in small space.

\paragraph{Partition scheme.} We maintain a partition of $T[1..\head]$ into disjoint segments stored in a linked list.
The length of every segment is a power of $2$, their lengths are nonincreasing as one moves to the right, and there is $M$
such that we
have between $A$ and $B$ segments of length $2^{\ell}$ for every $\ell=0,1,\ldots,M-1$, and between $1$ and $B$ segments of
length $2^{M}$. $A$ and $B$ are constants to be specified later. After increasing $h$ by one, a new segment of
length $2^{0}$ appears, then we possibly take two adjacent segments of length $2^{\ell}$ such that the
segment on their left (if any) is longer, and merge them into one segment of length $2^{\ell+1}$. We call this a partition scheme, as there is some
flexibility as to when the merging happens.

\begin{restatable}{lemma}{restatablepartition}
\label{lemma:partition}
There is a partition scheme with $A=3$ and $B=5$, which guarantees that after adding a new segment of length $2^{0}$ we can
merge in $\bigo(1)$ time at most one pair of adjacent segments of length $2^{\ell}$, such that there are $3$ segments
of the same length $2^{\ell}$ on their right.
\end{restatable}

\begin{proof}
This is a simple example of the recursive slow-down method of Kaplan and Tarjan~\cite{KaplanSlowdown}.
Let $2^{a_{1}}, 2^{a_{2}}, 2^{a_{3}}, \ldots$ be the lengths of the segments in the current partition, where $a_{1}\leq a_{2}\leq a_{3}\leq\ldots$.
We group together all segments with the same length, and denote the number of segments of length $2^{\ell}$ by $c_{\ell}$.
We will show how to maintain $c_{\ell}\in\{3,4,5\}$ for every $\ell=0,1,2,\ldots,M-1$ and $c_{M}\in\{1,2,3,4,5\}$,
where $2^{M}$ is the maximum length of a segment
in the current partition. To this end, we will keep the following invariant: if $c_{i}=5$ then there exists $j\in\{0,1,\ldots,i-1\}$ such that $c_{j}=3$ 
and $c_{j+1}=\ldots=c_{i-2}=c_{i-1}=4$. We call such partition valid.

We must show that, given a valid partition of $T[1..\head]$, we
can construct in $\bigo(1)$ time a valid partition of $T[1..\head+1]$. We start with creating a new segment of length $2^{0}$ and adding it to the
previous partition, which increases $c_{0}$ by one. Now there are two cases.

\begin{description}
\item[$c_{0}=5$] We merge two (leftmost) segments of length $2^{0}$ into a segment of length $2^{1}$, or in other words
we decrease $c_{0}$ by two ($c_0$ is now equal to 3) and increase $c_{1}$ by one. Because the initial value of $c_{0}$ was $4$, the only way the invariant could have been
broken is that $c_{1}$ was $3$, $c_{2}=\ldots=c_{i-1}=4$ and $c_{i}=5$ for some $i\geq 3$. But then the new $c_{1}$ becomes $4$,
and all $c_{2}, c_{3}, \ldots, c_{i-1}$ are now $4$, so the invariant holds.
\item[$c_{0}=4$] If there is $i$ such that $c_1=\ldots=c_{i-2}=c_{i-1}=4$ and $c_{i}=5$, then we merge the two
(leftmost) segments of length $2^{i}$ into a segment of length $2^{i+1}$, which decreases $c_{i}$ by two and increases $c_{i+1}$ by one.
As in the previous case, the only way the invariant could have been broken is that $c_{i+1}$ was $3$, $c_{i+2}=c_{i+3}=\ldots=c_{j-1}=4$
and $c_{j}=5$ for some $j\geq i+2$. Then $c_{i}$ becomes $3$, and all $c_{i+1}, c_{i+2},\ldots,c_{j-1}$ are now $4$, so the invariant holds.
\end{description}

To implement the update, we group together all consecutive $i$'s with the same value of $c_{i}$. In other words, we store a list of lists of segments.
This allows us to find $i$ from the second case in $\bigo(1)$ time.
\end{proof}

Instead of storing every alive $P(c)$ explicitly, for every segment we group together all alive processes such that $c$ lies inside.
We need the following result, which follows from a definition of a palindrome, see~\cite{Apostolico}.

\begin{lemma}
\label{lemma:periodicity}
If $c<c'$, $c'-c\leq 2^{\ell}$ and $R(c),R(c')\geq 2^{\ell}$, then $2(c'-c)$ is a period of $T[(c-2^{\ell})..(c'+2^{\ell}-1)]$.
\end{lemma}

The intuition is that in a segment of length $2^{\ell}$ either there are at most $4$ alive processes which can be kept explicitly,
or there are at least $5$ of them and the whole segment is periodic with period at 
most $2^{\ell-1}$. Hence for every segment we store either a sparse or a dense description, depending (roughly) on the periodicity of
the corresponding fragment.

\begin{description}
\item[Sparse description.] We explicitly store a list of all processes inside the segment, which can be potentially still alive.
We guarantee that there are at most $4$ processes on that list, and that if a process is not on the list, it is surely dead. We do not guarantee
that all processes on the list are still alive, but whenever we run one of them and the check fails, we declare it dead and remove from the list.
The processes currently on the list are called relevant.

\begin{figure}[t]
\includegraphics[width=\textwidth]{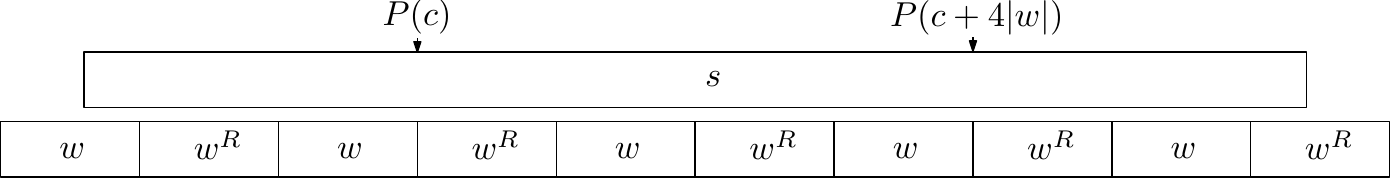}
\caption{Alive processes inside $s$ with a dense description are of the form $P(c+\alpha|w|)$.}
\label{fig:dense}
\end{figure}

\item[Dense description.] We guarantee that there exists a word $w$ such that $|w|\leq\frac{1}{4} 2^{\ell}$ for which the whole segment 
of length $2^{\ell}$
is a subword of $(ww^{R})^{\infty}$, see Fig.~\ref{fig:dense}. Denoting the segment by $s$, this implies that $\per(s)\leq\frac{1}{2}|s|$
and $s$ has a palindromic subword of length $\per(s)$. In such a case we store a~multiple of the period, denoted by $\kpers = k \per(s) \leq\frac{1}{2}2^{\ell}$,
such that the only alive processes inside the segment are of the form $P(c+\alpha \frac{\kpers}{2})$ for $\alpha\geq 0$, where $T[c..(c+\kpers-1)]$ is an even palindrome fully within the segment.
(We do not require that all such processes are still alive.) We store $c$ and $\kpers$, which is enough to run any relevant process
inside $s$ in $\bigo(1)$ time, where relevant means of the form $P(c+\alpha\frac{\kpers}{2})$.
No other process inside $s$ can be alive.
\end{description}

We use Lemma~\ref{lemma:partition} to maintain a partition of $T[1..\head]$ into segments. The description of every segment requires
just $\bigo(1)$ space, making the total additional space complexity $\bigo(\log n)$. 
After reading $T[\head]$ we create a sparse
description of the new segment of length $2^{0}$ and then need to merge at most one pair of adjacent segments. After having updated the
partition, we can simply run all relevant processes.
Therefore, now we need to show how to merge a pair of adjacent segments $s$ and $s'$ of length
$2^{\ell}$ as to obtain a new segment $ss'$. % There are two cases depending on how the segments $s$ and $s'$ are described.
If their descriptions are sparse, we merge the lists of $s$ and $s'$ and either get at most $4$ processes, which constitute a valid sparse description,
or at least $5$ processes $P(c_{1}),P(c_{2}),\ldots,P(c_{5})$. The following observation follows from the Lemma~\ref{lemma:partition}.

\begin{observation}
\label{obs:lifespan}
When a segment of length $2^{\ell}$ is being created, the number of already seen characters on its right is at most
$3\cdot 2^{\ell-1}+5(2^{\ell-1}-1)=2^{\ell+2}-5$.
When it is being destroyed, there are at least $3(2^{\ell+1}-1)$ of them.
\end{observation}

Hence any $P(c_{i})$ could have been run everywhere in the interval $[c_{i}+2^{\ell}+2^{\ell+2}-5,c_{i}+3(2^{\ell+1}-1)]$,
and by Lemma~\ref{lemma:delay} had at least one landmark available in that interval,
so its radius must be at least $2^{\ell}+2^{\ell+2}-4\geq 2^{\ell+1}$. Therefore, we have a list of $5$ processes inside a segment of length
$2^{\ell+1}$, all of which have radii at least $2^{\ell+1}$ ($\ell=0$ must be considered separately).
This suffices to construct a dense description by the following lemma.

\begin{restatable}{lemma}{restatablemerge}
\label{lemma:merge}
Given a list of $m\geq 5$ processes $P(c_{1}),P(c_{2}),\ldots,P(c_{m})$ inside a segment of length $2^{\ell}$, such that their radii are all at least $2^{\ell}$
and no other process inside is alive, we can construct in $\bigo(m+\log n)$ time a dense description.
%of the segment.
\end{restatable}

\begin{proof}
We rearrange the processes so that $c_{1}<c_{2}<\ldots<c_{m}$ and define $\Delta_{i}=c_{i+1}-c_{i}$. Every $2\Delta_{i}$ is a period of the segment
by Lemma~\ref{lemma:periodicity}. We claim that by the periodicity lemma also $\gcd(2\Delta_{1},2\Delta_{2},\ldots,2\Delta_{m-1})$ is a period of the 
segment. This can be seen by the following reasoning: if the radii at $c<c'<c''$ are all at least $2^{\ell}$, $c''-c\leq 2^{\ell}$, $2d \divides 2(c'-c)$ is a period of
$T[(c-2^{\ell})..(c'+2^{\ell}-1)]$ and $2d'\divides 2(c''-c')$ is a period of $T[(c'-2^{\ell})..(c''+2^{\ell}-1)]$, then by the periodicity lemma $2\gcd(d,d')$ is a
period of  the whole $T[(c-2^{\ell})..(c''+2^{\ell}-1)]$. Then by induction $\kpers=2\gcd(\Delta_{1},\Delta_{2},\ldots,\Delta_{m-1})$ is a period
of $T[(c_{1}-2^{\ell})..(c_{k}+2^{\ell}-1)]$, which contains the whole segment inside. Because $\Delta_{1}+\Delta_{2}+\ldots+\Delta_{m}\leq 2^{\ell}$
and $m\geq 5$,
$\Delta_{i} \leq \frac{1}{4}2^{\ell}$ for at least one $i$, so $\kpers\leq\frac{1}{2}2^{\ell}$ and consequently $\kpers$ must be a multiple of $\per(s)$.

Now we can construct a dense description. We compute $\kpers$ in $\bigo(m+\log n)$ with $m-1$ applications of the Euclidean algorithm,
and set $c=c_{1}$.
Because $\kpers\divides2\Delta_{i}$ for every $i$, all $c_{i}$ are of the form $c+\alpha\frac{\kpers}{2}$. Furthermore, because $\kpers\leq \min(\Delta_{1},\Delta_{2})$, 
$c+\kpers\leq c_{2}$, so $T[c..(c+\kpers-1)]$ is fully within the segment. Finally, we must argue that $T[c..(c+\kpers-1)]$ is an even palindrome.
First observe that $T[(c_{3}-\kpers)..(c_{3}+\kpers)]$ lies fully within the segment, and consider two cases.
\begin{itemize}
\item If $c_{3}=c+\alpha \kpers$, then $T[(c_{3}-\kpers)..(c_{3}+\kpers)]=T[c..(\kpers-1)]^{2}$. Because the palindromic radius at $c_{3}$ is at least $2^{\ell}\geq \kpers$,
$T[c..(c+\kpers-1)]$ is a palindrome.
\item If $c_{3}=c+\frac{\kpers}{2}+\alpha p$, then $T[(c_{3}-\frac{\kpers}{2})..(c_{3}+\frac{\kpers}{2})]=T[c..(c+\kpers-1)]$. Because the palindromic radius at $c_{3}$ is at least $2^{\ell}\geq \frac{\kpers}{2}$, $T[c..(c+\kpers-1)]$ is a palindrome.\qedhere
\end{itemize}
\end{proof}

This settles the situation when both descriptions are sparse. Before we move to the remaining case, we need an additional tool. If a
description of a segment is dense, we maintain some additional information about the processes inside. Informally, we would like to know which
of them are still alive, but of course we cannot afford to explicitly maintain such information. We can only afford to store a short buffer, 
where we keep information about a few most recently run processes. Formally, the buffer is a list of processes $P(c)$ together with their
corresponding values of $R(c)$. We do not require that $P(c)$ is still alive, so it might have happened that it has been run again after reading $T[\head']$
with $h<h'$, but the more recent run was unsuccessful. The buffer is updated whenever we successfully run a process $P(c)$ inside the segment.
There either $P(c)$ was in the buffer, so we move it to the front and update the corresponding $R(c)$, and otherwise
we prepend it to the buffer together with the current $R(c)$, and if the length of the buffer is now $6$ we remove the last element from there.
 Hence the buffer is of length at most $5$. A less trivial consequence is as follows.

\begin{restatable}{lemma}{restatableconvergence}
\label{lemma:convergence}
If a segment with dense description of length $2^{\ell}$ is being destroyed while at most $4$ processes in its buffer have radii at least $2^{\ell+1}$, 
then no other process inside the segment can be still alive.
\end{restatable}

\begin{proof}
By Observation~\ref{obs:lifespan} and how we process segments with dense descriptions, any $P(c)$ which might be still alive could have been run everywhere
in the interval $[c+2^{\ell}+2^{\ell+2}-5,c+3(2^{\ell+1}-1)]$, and by Lemma~\ref{lemma:delay} it had at least one
landmark available in that interval. Also, whenever we run any $P(c)$ inside the segment in the interval $[c+2^{\ell}+2^{\ell+2}-5,\infty)$, and 
it succeeds, $R(c)$ is set to at least $2^{\ell+1}$ (except when $\ell=0$, but then there is just one process inside the segment, so the buffer
surely contains it). Therefore, if the buffer contains at most $4$ processes with radii at least $2^{\ell+1}$, any $P(c)$ such that $R(c)\geq 2^{\ell+1}$
is stored in the buffer, and no other process can be still alive.
\end{proof}

If at least one description is dense, by applying Lemma~\ref{lemma:convergence} to $s$ (if its description is dense) or $s'$ (if its description is dense), we either get that one of these
segments contains at least $5$ processes with radii at least $2^{\ell+1}$ in its buffer, or we get a list of at most $4$ potentially still alive processes
inside each segment.
In the latter case we concatenate the lists to get a list of at most $8$ processes inside $ss'$ such that all other processes
inside are dead. If the list contains at most $4$ processes, we construct a sparse description of $ss'$, and otherwise we apply Lemma~\ref{lemma:merge} to construct a dense description of $ss'$ in $\bigo(\log n)$ time.
In the former case we get a list of between $5$ and $10$ alive processes $P(c_{1}),P(c_{2}),\ldots,P(c_{m})$ inside $ss'$. It might be the case that 
there are also some other processes inside the segment which are still alive, but they are not stored in the buffer of the corresponding segment. Nevertheless,
proceeding as in the proof of Lemma~\ref{lemma:merge} we can compute in $\bigo(\log n)$ time $\kpers$ and $c$ such that $T[c..(c+\kpers-1)]$ is an even
palindrome fully within $ss'$, all $c_{i}$ are of the form $c+\alpha\frac{\kpers}{2}$, and $\kpers\leq\frac{1}{2}|ss'|$ is a period of $ss'$. This is not a valid 
dense description yet, as $s$ or $s'$ (or both) might have dense descriptions, and we cannot guarantee that all alive processes 
there are of the form $c+\alpha\frac{\kpers}{2}$.

Consider the case when $s$ has a dense description, meaning that we have $\kpers'$ and $c'$ such $T[c'..(c'+\kpers'-1)]$ is an even palindrome fully within $s$,
all alive processes there are of the form $c'+\alpha\frac{\kpers'}{2}$, and $\kpers'\leq \frac{1}{2}|s|$ is a period of $s$. If $\kpers'\divides\kpers$ there is nothing to do.
Otherwise, because the list $P(c_{1}),P(c_{2}),\ldots,P(c_{m})$ contains at least $5$ processes inside $s$ we have $\kpers\leq\frac{1}{2}|s|$ and
by the periodicity lemma $\gcd(\kpers,\kpers')$ is a period of $s$. Then $\gcd(\kpers,\kpers')$ must be actually a period of the whole $ss'$. 
Now we claim that $\kpers$ can be, in fact, replaced by $\gcd(\kpers,\kpers')$. This is because if a power of a word is a palindrome, the word itself must be a palindrome,
so $T[c'..(c'+\gcd(\kpers,\kpers')-1)]$ is an even palindrome.

The case when $s'$ has a dense description, or both $s$ and $s'$ have dense descriptions, can be dealt with similarly.

\begin{theorem}
\label{memory_efficient}
For any scheduling scheme with $b_{\level}\geq 12$ for all $\level\leq L$, descriptions of all segments in the current
partition of $T[1..\head]$ can be maintained in $\bigo(\log n)$ space and $\bigo(\log n)$ time plus the
time to run all relevant processes.
\end{theorem}

\section{Time-efficient algorithm}
\label{section:time2}
The simulation from the previous section was space-efficient, but not time-efficient yet, because there might be segments with
dense descriptions and small periods, which in turn requires running many relevant processes. 
This is the only reason the time to process $T[\head]$ might exceed $\bigo(\log n)$, as merging at most one pair of segments and
running the processes in all segments with sparse descriptions takes just $\bigo(\log n)$ time. In this section we show how to simulate
running all relevant processes in a segment with dense description in $\bigo(1)$ time.

Consider a dense description of a segment $s$. Recall that it consists of $c$ and $\kpers$, such that $T[c..(c+\kpers-1)]$ is an even palindrome
and $\kpers$ is a period of the whole segment, and we want to run
all processes $P(c')$ inside $s$ of the form $c'=c+\alpha\frac{\kpers}{2}$, where $2c-\head-2$ is a landmark.
We can construct and run all relevant processes in $\bigo(1)$ time each, but there might be many of them. However, there are
only two consequences of running such a $P(c')$: we might update the final answer, and we might also store it in the buffer (or move it to the front there).
Therefore, if we can guarantee that a particular $P(c')$ will fail anyway, we can avoid running it altogether. Similarly, if we can guarantee that many
processes $P(c')$ will succeed, it is enough to run just the $5$ leftmost of them.
We will build on these observations to simulate running all processes of such form in a single segment with a dense description in $\bigo(1)$ total time.
This is the most technical part, so we start with an overview.

\paragraph{Overview.}   We start with observing in Lemma~\ref{lem:few levels} that, when considering such a segment, just a constant number of \emph{associated landmark
levels} needs to be considered. Then we analyze which relevant processes inside a segment should be run because of a landmark on level $\level$.
After some basic arithmetical manipulation, we get a succinct description of all such values of $c'$. To avoid considering all of them, which
might be too costly, we apply two lemmas characterizing the structure of palindromes in a sufficiently periodic fragment of the text,
described in Lemma~\ref{lemma:structure1} and Lemma~\ref{lemma:structure2} (these observations go back to~\cite{Apostolico}, but we need
a slightly different formulation). To apply them, we need to compute how far the periodicity of a segment with a dense description continues
to the left and to the right. To this end, we relax the notion of landmarks, introducing the so-called \emph{ghost landmarks}, which allow us to
operate on a longer suffix of the already seen $T[1..\head]$. Then, using the ghost landmarks, we binary search to compute how far
the periodicity extends, and apply the structural results to isolate at most $5$ relevant processes, which should be run as to guarantee
the correctness. To achieve the final complexity of $\bigo(\log n)$  to process $T[h]$, we precompute how far the
periodicity continues when creating the segment, and then maintain this information in $\bigo(1)$ time.

%This gives us the claimed the complexity for simulating any scheme efficiently.

\paragraph{Associated landmark levels.} Consider a segment $s$. If, for some $c$ inside $s$, $2c-\head-2$ is a~landmark strictly strictly on
level $\level$ at $h$, we say that $\level$ is a~landmark level associated to $s$.

\begin{restatable}{lemma}{restatablefewlevels}
\label{lem:few levels}
There are at most $4$ landmark levels associated to a single segment, and they can be all determined
$\bigo(\log n)$ time.
\end{restatable}

\begin{proof}
Consider a segment $s$ of length $2^{\ell}$ and any $c$ inside. By Observation~\ref{obs:lifespan}, when the segment is being created
by merging two segments of length $2^{\ell-1}$ we have $h-c \geq  3(2^{\ell}-1)$. Similarly, when the segment is being destroyed
by merging with an adjacent segment of length $2^{\ell}$ to form a segment of length $2^{\ell+1}$ we have $h-c < 2^{\ell+1}+2^{\ell+3}-5=
5(2^{\ell+1}-1)$. Consequently, we can bound $2(h-c+1)$, which is the number of already seen characters on the right of $2c-h-2$, as follows:
\begin{eqnarray*}
2(h-c+1) &<& 10\cdot 2^{\ell+1}-8 \\
2(h-c+1) &\geq&  6\cdot 2^{\ell}-4
\end{eqnarray*}
If $2c-h-2$ is a landmark strictly on level $\level<L$, then the number of already seen characters on its right belongs to
$[b_{0}\cdot 2^{\level-1},b_{0}\cdot 2^{\level})$. Bounding the number of different
landmark levels associated to $s$ requires counting $\level<L$ such that 
$[b_{0}\cdot 2^{\level-1},b_{0}\cdot 2^{\level}) \cap [6\cdot 2^{\ell}-4,10\cdot 2^{\ell+1}-8) \neq \emptyset$.
The condition translates into:
\begin{eqnarray*}
b_{0} \cdot 2^{\level-1} &\leq& 10\cdot 2^{\ell+1} - 8 - 1 \\
b_{0} \cdot 2^{\level}-1 &\geq& 6\cdot 2^{\ell}-4
\end{eqnarray*}
which is equivalent to $2^{\level}\cdot b_{0} \in [6\cdot 2^{\ell}-3,40\cdot 2^{\ell}-18]$. If $\level=L$, the number of already seen characters on
the right is at least $b_{0}\cdot 2^{\level-1}$, so the condition becomes $2^{\level}\cdot b_{0}\leq 40\cdot 2^{\ell}-18$.
All in all, there are at most $4$ different possible values of $\level$.

Generating the landmark levels associated with a given segment can be done in $\bigo(\log n)$ time by performing the above calculation.
\end{proof}

Due to the above Lemma~\ref{lem:few levels}, to achieve the claimed $\bigo(\log n)$ time complexity for processing $T[h]$,
we only need to show how to run all relevant processes inside a segment with a dense description using landmarks on a particular level
$\level$ associated to that segment in $\bigo(1)$ time.

\paragraph{Relevant processes.} We need to consider all relevant processes $P(c')$, such that $2c'-h-2$ is a landmark on level $\level$, implying that
$2^{\level} \divides 2c'-h-2$. The condition is equivalent to:
\begin{equation}
\label{eq:alfa1}
\alpha \cdot \kpers = h+2-2c \pmod{2^{\level}}
\end{equation}
which, denoting $2^{\ell}=\gcd(\kpers,2^{\level})$, is in turn equivalent to:
\[ \alpha \cdot \frac{\kpers}{2^{\ell}} = \frac{h+2-2c}{2^{\ell}} \pmod{2^{\level-\ell}}\]
(unless $2^{\ell}$ does not divide $h+2-2c$, when no $c'$ needs to be considered), so by computing the multiplicative inverse we finally get
a base solution to \eqref{eq:alfa1}:
\[ \alpha_0 = \frac{h+2-2c}{2^{\ell}} \cdot \left(\frac{\kpers}{2^{\ell}}\right)^{-1} \pmod{2^{\level-\ell}}\]
and the general solution is:
\[ \alpha = \alpha_0 + t \cdot 2^{\level-\ell} \quad\text{for}\quad t \in \{\ldots,-1,0,1,\ldots\}\]
Thus we also get the solution to the original equation:
\begin{equation}
\label{eq:gform}
c' = c'_0 +  t \cdot \frac{\kpers}{2}2^{\level-\ell} = c'_0 + t \cdot \frac{1}{2} \lcm(2^\ell,\kpers) \qquad \textrm{ where } c'_0 = c+ \alpha_0 \frac{\kpers}{2}.
\end{equation}

Therefore, with a simple calculation we get a succinct description of all values of $c'$ which should be taken into the account.
Before we proceed further, let us comment on the complexity of the calculation. Since $\level$ is fixed, both values of
\[2^{\ell}=\gcd(\kpers,2^{\level}) \quad\text{and}\quad \left(\frac{\kpers}{2^{\ell}}\right)^{-1}\bmod{2^{L-\ell}}\]
can be computed in $\bigo(\log n)$ time when we create the segment and stored there.

The situation now is that we have a dense description of a segment, and want to run all processes $P(c')$ of the form \eqref{eq:gform}
inside the segment. Additionally, because we do not necessarily have all possible landmarks on level $\level$, just a few most recent,
we are interested only in sufficiently large $c'$. Observe, that we can analyze separately processes of the following two forms:
\begin{eqnarray}
P(c'_0 + t \cdot \lcm(2^\level,\kpers)) & \label{eq:4} \\
P(c''_0 + t \cdot \lcm(2^\level,\kpers)) &\quad\text{where}\quad c''_0 = c'_0 + \frac{1}{2}\lcm(2^\level,\kpers) \label{eq:5}
\end{eqnarray}
From now on we will only consider the former, as the whole reasoning still holds after replacing $c'_0$ by $c''_0$. We will also assume
that only $t\geq 0$ need to be considered, which can be  ensured by decreasing $c'_{0}$ by an appropriate multiple of
$\lcm(2^{\level},\kpers)$.

Because $\kpers$ is a period of the whole segment, $\lcm(2^{\level},\kpers)$ is its period as well, and furthermore we can assume that
$\lcm(2^{\level},\kpers)\leq\frac{1}{2}|s|$, as otherwise there are just at most two relevant processes to run.
Intuitively, knowing how far the period extends to the left and to the right allows us to restrict the number of processes to run
by an argument based on the combinatorial properties of palindromes.
While computing how far the period extends exactly is not possible in our setting, it can be approximated quite well using the landmarks.
First, we need to introduce the notion of ghost landmarks.

\paragraph{Ghost landmarks.} For every level of landmarks $\level$, we store $f_{\level} = 4 \cdot b_{\level}$ most recently seen
landmarks on level $\level$. We call them ghost landmarks on level $\level$. All ghost landmarks can be maintained in the
same manner as the regular landmarks, so storing them does not change the complexity of our algorithm.

\begin{lemma}
\label{lem:using ghosts}
If $\level$ is a landmark level associated to a segment $s$, then for any $c$ inside $s$ there exists at least one ghost landmark on
level $\level$ in $T[1..(2c-\head-2)]$.
\end{lemma}

\begin{proof}
Consider a segment $s$ of length $2^{\ell}$. By Observation~\ref{obs:lifespan}, the number of already seen characters on the right
of $s$ when it is being created is at least $3(2^{\ell}-1)$. Let $c'$ be any position inside $s$
causing $\level$ to be associated to $s$, \ie, $2^{\level} \divides 2c'-\head'-2$, and denote
$x' = 2c'-\head'-2$. Notice that $\head'$ might be either smaller or larger than the current $\head$. Because $c'$ is a landmark
on level $\level$ at $\head'$, we have that $2^\level \cdot b_\level \ge \head' - x' = 2(\head' - c')  + 2 \ge 2 + 6(2^{\ell}-1)$.

Now consider any $c$ inside $s$ and denote $x = 2c - \head -2$. Since $c$ and $c'$ both belong to the same segment of length $2^{\ell}$,
$c'-c \le 2^{\ell}-1$. Applying Observation~\ref{obs:lifespan} again, we also get that
$\head-\head' \le 5(2^{\ell+1}-1)-1 - 3(2^{\ell}-1) = 7\cdot 2^{\ell} - 3$.

Thus the number of already seen characters on the right of $x$ can be bounded as follows:
\[\head - x = 2\head - 2c + 2 \le (2\head' - 2c' + 2) + 2(2^{\ell}-1) + 14 \cdot 2^{\ell} - 6.\]
Because $\head' - x' = 2(\head' - c')  + 2 \ge 2 + 6(2^{\ell}-1) $, we have $16\cdot 2^{\ell} - \frac{32}{3} \le \frac{8}{3}\cdot (\head' -x')$, so
the above bound can be rewritten as:
\[\head - x \le \frac{11}{3} \cdot (\head' -x') +\frac{8}{3}\le \frac{11}{3}b_\level\cdot 2^\level +\frac{8}{3} \le (4 b_\level-1) \cdot 2^\level\]
where the last inequality holds because $b_\level \ge 12$.
Since the leftmost ghost landmark on level $\level$ has at least $(4 b_\level-1) \cdot 2^\level$ already seen characters on its right,
by the above calculation it must be on the left of $x=2c-\head-2$ as claimed.
\end{proof}

Now going back to approximating how far the period extends to the left and to the right, we proceed as follows.
We choose $w$ of length $\lcm(2^{\level},\kpers)$ starting at $T[c'_0]$. Because we have adjusted $c'_{0}$ so that only $t\geq 0$ need to
be considered and $|w| \leq \frac{1}{2}|s|$, we can assume that $w$ is fully within the segment.
We know that the whole segment can be covered by repeating $w$ to the left and to the right (where, possibly, the last repetition is
a suffix or a prefix of $w$, respectively), and would like to figure out how far we can continue that until we hit either a boundary of the already seen
$T[1..\head]$, or a subword of length $|w|$ which is different than $w$. This can be approximated quite well using the ghost landmarks,
if $w$ repeats at least twice.

\begin{lemma}
For any $w$ such that $T[i..(i+2|w|-1)]=w^{2}$, $2^{\level} \divides |w|$, and $T[1..(i-1)]$ contains at least one ghost landmark on level
$\level$, we can compute in $\bigo(\log h)$ time $r \geq 2$
such that $T[i..(i+r |w|-1)]=w^{r}$ and either $i+(r+2) |w| > h$ or $T[i..(i+(r+2) |w|-1)] \neq w^{r+2}$.
\label{lemma:extend}
\end{lemma}

\begin{proof}
By the assumption about ghost landmark on level $\level$, we can access any $\bhash(2^{\level}\cdot j)$ with
$j\geq\left\lfloor \frac{i-1}{2^{\level}}\right\rfloor$ in $\bigo(1)$ time.
Hence if we are lucky and $i=2^{\level}\cdot j+1$, we can compute $\hash(T[(i+\alpha|w|)..(i+\beta|w|-1)]$ for any $0\leq\alpha\leq\beta$
in $\bigo(1)$ time, which allows us to binary search for $r$ in $\bigo(\log h)$ time. In more
detail, to check if $T[i..(i+r |w|-1)]=w^{r}$ we check if $|w|$ is a~period of $T[i..(i+r |w|-1)$,
which can be done by comparing $\hash(T[(i+|w|)..(i+r |w|-1)])$ and
$\hash(T[i..(i+(r-1)|w|-1)])$.

\begin{figure}[t]
\includegraphics[width=\textwidth]{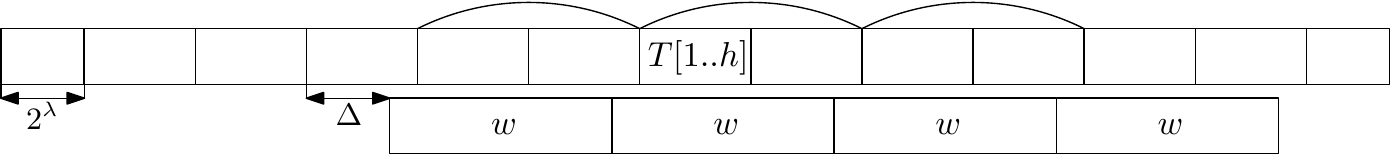}
\caption{A number of repetitions of $w$ such that $2^{\level} \divides |w|$ implies that $|w|$ is a~period of a
certain full fragment between two ghost landmarks.} 
\label{fig:extend}
\end{figure}

In the general case, let $i = 2^{\level}\cdot j+1+\Delta$, where $\Delta\in [0,2^{\level})$.
If $T[i..(i+r |w|-1)]=w^{r}$ and $r\geq 2$, then $|w|$ is a~period of $T[(2^{\level}\cdot j+2^{\level}+1)..(2^{\level}\cdot j +2^{\level}+\alpha |w|)]$,
see Fig.~\ref{fig:extend}, where $\alpha= r-1$. In the other direction, if $|w|$ is a~period of
$T[(2^{\level}\cdot j+2^{\level}+1)..(2^{\level}\cdot j +2^{\level}+\alpha |w|)]$ and $r \geq 2$, then $r \geq \alpha$. (The assumption that
$r \geq 2$ is crucial.) Hence we can determine the largest $\alpha$ such that
$|w|$ is a~period of $T[(2^{\level}\cdot j+2^{\level}+1)..(2^{\level}\cdot j +2^{\level}+\alpha |w|)]$ in $\bigo(\log h)$ time using ghost landmarks
on level $\level$, and then simply return $\alpha$, which guarantees $r \in \{\alpha,\alpha+1\}$.
\end{proof}

\begin{lemma}
For any $w$ such that $T[i..(i+2|w|-1)]=w^{2}$, $2^{\level} \divides |w|$, and $T[1..(2i-\head-2)]$ contains at least one ghost landmark on level
$\level$, we can compute in $\bigo(\log h)$ time $\ell \geq 0$
such that $T[(i-\ell |w|)..(i-1)]=w^{\ell}$ and either $i-(\ell+2) |w|< 2i-h-2$ or $T[(i-(\ell+2) |w|)..(i-1)] \neq w^{\ell+2}$.
\label{lemma:extend2}
\end{lemma}

\begin{proof}
The proof will be very similar to the proof of Lemma~\ref{lemma:extend}, except that we have to take into the account the fact that while there might
be many more repetitions of $w$ to the left, we might not have enough ghost landmarks on level $\level$ to detect them.

Let $i = 2^{\level}\cdot j+1+\Delta$, where $\Delta\in [0,2^{\level})$.
If $T[(i-\ell |w|)..(i-1)]=w^{\ell}$, then $|w|$ is a period of $T[(2^{\level}\cdot j-\alpha |w|+2^{\level}+1)..(2^{\level}\cdot j+|w|+2^{\level})]$,
where $\alpha= \ell$. In the other direction, if $|w|$ is a period of
$T[(2^{\level}\cdot j-\alpha |w|+2^{\level}+1)..(2^{\level}\cdot j+|w|+2^{\level})]$, then $\ell \geq \alpha-1$.
So we only need to binary search for the largest $\alpha$ such that $|w|$ is a period of
$T[(2^{\level}\cdot j-\alpha |w|+2^{\level}+1)..(2^{\level}\cdot j+|w|+2^{\level})]$ and return $\max(0,\alpha-1)$.
The remaining difficulty is that $2^{\level}\cdot j-\alpha |w|+2^{\level}$ might lie too far on the left to be a ghost landmark on level $\level$,
so the binary search needs to be slightly modified.
We first choose the largest $\alpha_{0}$ such that $2^{\level}\cdot j-\alpha_{0} |w|+2^{\level}$ is a ghost landmark on level $\level$. There
are two possibilities.
\begin{enumerate}
\item $|w|$ is a period of $T[(2^{\level}\cdot j-\alpha_{0} |w|+2^{\level}+1)..(2^{\level}\cdot j+|w|+2^{\level})]$, then the largest
$\alpha$ might exceed $\alpha_{0}$. But we can return $\ell=\max(0,\alpha_{0}-1)$, because then
$i-(\ell+1) |w|\leq i-\alpha_{0}|w|$, and the choice of $\alpha_{0}$ and the assumption, by Lemma~\ref{lem:using ghosts}, implies
$i-(\alpha_{0}+1)|w|<2i-\head-2$,
so $i-(\ell+2)|w| < 2i-\head-2$.
\item Otherwise, we binary search over all $\alpha\leq \alpha_{0}$, and return $\max(0,\alpha-1)$.
\end{enumerate}
We can binary search for $\alpha_{0}$ in $\bigo(\log h)$ time, so the total time is $\bigo(\log h)$.
\end{proof}

We apply Lemma~\ref{lemma:extend} and Lemma~\ref{lemma:extend2} to
approximate how many times $w$ repeats on its right and on its left in $T[(2c'_{0}-\head-2)..(c'_{0}-1)]$ with accuracy $1$,
assuming that $w^{2}$ occurs at $T[c'_{0}]$.
Notice that there might be many more repetitions to the left in the whole $T[1..(c'_{0}-1)]$, but Lemma~\ref{lemma:extend2} does not
allow us to detect all of them.
Now the crucial insight is that even though we do now know the exact number of repetitions, we can iterate through the at most $4$ possible combinations of
the number of of repetitions to the left and to the right right, and the additionally consider the possibility that there is only a single
occurrence of $w$ in the segment. Hence we need to iterate through $5$ possibilities in total.
For each such combination, we will restrict the number of processes which should be run, therefore by running the processes
determined for each of these combinations we will not lose the correctness.
Hence from now on we assume that we know the exact number of repetitions of $w$ to the left and to the right.

We need the following two simple structural results, which allow us to bound the palindromic radius in a sufficiently periodic subword of the
text. A similar (in spirit) argument appeared already in~\cite{Apostolico}, but we need a slightly different formulation. We say that a palindrome
centered at $c$ reaches $h$ if $R(c)\geq h-c+1$.

\begin{lemma}
\label{lemma:structure1}
Consider $uw^{k}v$ starting at position $i$ in $T[1..\head]$, where $|u|=|w|=|v|$, $w$ is a palindrome, and $u,v\neq w$.
For any $\alpha\in\{1,2,\ldots,k\}$, if the palindrome centered at $i+\alpha|u|$ reaches $h$ then $\alpha=\frac{k}{2}+1$.
\end{lemma}

\begin{proof}
Take any $\alpha\in\{1,2,\ldots,k\}$. For a palindrome centered at $i+\alpha|u|$ to reach $c$, $R(i+\alpha|u|)$ must be at least
$\min(\alpha-1,k+1-\alpha)|u|$. But then either $u=w^{R}$ or $v=w^{R}$, which is a contradiction.
\end{proof}

\begin{lemma}
\label{lemma:structure2}
Consider $uw^{k}$ starting at position $i$ in $T[1..\head]$, where $|u|=|w|$, $w$ is a palindrome, $u\neq w$, and $h-i-(k+1)|u|+1<|w|$.
For any $\alpha\in\{1,2,\ldots,\lceil\frac{k}{2}\rceil\}$, the palindrome centered at $i+\alpha |u|$ cannot reach $h$.
Additionally, either all palindromes centered at $i+\alpha |u|$ with $\alpha\in\{\lceil\frac{k}{2}\rceil+1,\ldots,k\}$ reach $h$, or none of them do.
\end{lemma}

\begin{figure}[t]
\includegraphics[width=\textwidth]{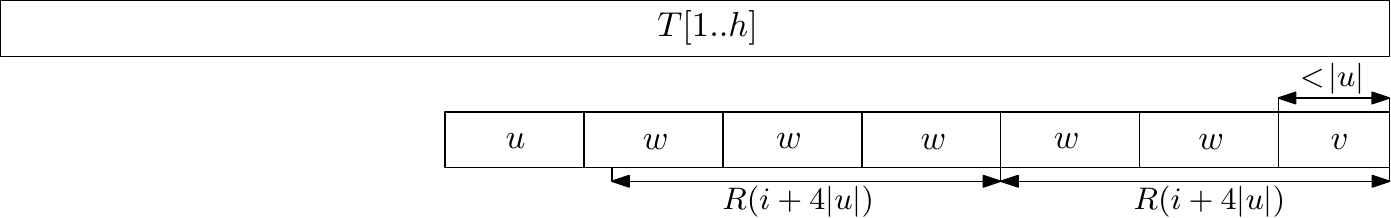}
\caption{All palindromes centered at positions $i+\alpha|u|$ with $\alpha\in\{\lceil\frac{k}{2}\rceil+1,\ldots,k\}$ reach $h$ if $v$ is a~prefix of $w$.}
\label{fig:structure2}
\end{figure}

\begin{proof}
Take any $\alpha\in\{1,2,\ldots,k\}$. If $\alpha\leq\lceil\frac{k}{2}\rceil$, then because $u\neq w$ the radius at $i+\alpha |u|$ is too small
for the palindrome centered at $i+\alpha|u|$ to reach $h$.
Otherwise, let $T[i..\head]=uw^{k}v$, where $|v|<|u|$ because $h-i-(k+1)|u|+1$, see
Fig.~\ref{fig:structure2}. Now either $v$ is not a prefix
of $w$, and we actually get the situation from Lemma~\ref{lemma:structure1}, so only $\alpha=\frac{k}{2}+1$ can possibly correspond
to a palindrome reaching $h$, or $v$
is a~prefix of $w$, and for all $\alpha\geq\lceil\frac{k}{2}+1\rceil$ the palindrome centered at $i+\alpha|u|$ reaches $h$.
\end{proof}

Recall that we want to run all $P(c'_0+t|w|)$ inside the segment with $t \geq 0$, and our $w$ starts at $T[c'_{0}]$.
We know that $w$ repeats $\ell$ times to the left in $T[(2c'_{0}-\head-2)..(c'_{0}-1)]$ and $r$ times to the right till the end of the
already seen $T[1..\head]$. The actual number of repetitions of $w$ to the left in the whole $T1..(c'_{0}-1)]$, denoted $\ell'$,
might be larger than $\ell$. By Lemma~\ref{lemma:structure1} and Lemma~\ref{lemma:structure2}, either all processes of the form
$P(c'_0+(-\ell'+\alpha)|w|)$ with $\alpha \geq \lceil\frac{\ell'+r}{2}\rceil+1$ will succeed, or just the one with
$\alpha = \frac{\ell'+r}{2}+1$ will succeed. Because the size of the buffer is $5$, we only need to ensure that the $5$ leftmost
processes which will succeed are run.
To guarantee this, we run all processes of the form $P(c'_{0}+(\max(-\ell+\lceil\frac{\ell+r}{2}\rceil +1,0) + x)|w|)$ for $x=0,1,2,3,4$
which are still inside the segment.
This is correct, as following two cases show.

\begin{enumerate}
\item The process $P(c'_{0}+(-\ell'+\alpha)|w|)$ with $\alpha = \frac{\ell'+r}{2}+1$ is on the left of the segment, so either all or none
processes of such form in the segment are alive.
\item The process $P(c'_{0}+(-\ell'+\alpha)|w|)$ with $\alpha = \frac{\ell'+r}{2}+1$ is inside the segment, so $\ell'$ cannot be too large.
More precisely, $\ell' \leq r$, and consequently $\ell=\ell'$.
\end{enumerate}

We run a constant number of processes, each of them in $\bigo(1)$ time, but to ensure that every segment is processed in such
complexity, we also need to remove the binary search used to approximate how
many times $w$ can be repeated to the left and to the right.

Recall that $|w|=\lcm(2^{\level},\kpers)$, $w$ starts at $T[c'_0]$ and lies
fully within a segment $s$, and furthermore $|w|$ is a period of the whole $s$. As mentioned before, we can also assume that $|w| \leq \frac{1}{2}|s|$,
as otherwise there are at most two processes which might need to be run. We can compute how many times $w$ can be repeated to
its left (or rather approximate this value as described in Lemma~\ref{lemma:extend2}) when the segment is created, as the result
does not depend on the current value of $h$.
Similarly, we can compute how many times it can be repeated to the right when we create the segment, but here the important
difference is that we might continue till the very end of the current $T[1..\head]$, \ie, the next copy of $w$ might extend beyond the
current prefix $T[1..\head]$. It can be seen that in such a case the next time we need to deal with the same segment, at most one
additional copy of $w$ fits inside $T[1..\head]$. This happens because the segment is relevant when $2^{\ell} \divides h+2-2c$, and
$|w|\geq 2^{\ell}$. Therefore, the number of times $w$ repeats to the right can be maintained in $\bigo(1)$ time.

\begin{theorem}
\label{thm:time_efficient}
Any scheduling scheme with $b_{\level}\geq 12$ for all $\level\leq L$
can be simulated using $\bigo(\log n)$ additional space on the top of the
space taken by the landmarks and $\bigo(\log n)$ time to process $T[\head]$.
\end{theorem}

\section{Lower bounds}
\label{section:lowerbounds}

In this section we use Yao's minimax principle~\cite{Yao77} to prove
lower bounds on the space complexity of computing the largest radius of a palindrome in a word
of length $n$ over an alphabet $\Sigma$ in the streaming model. We denote this problem
by \palin.

\begin{theorem}[Yao's minimax principle for randomized algorithms]
\label{yao}
Let $\mathcal{X}$ be the set of inputs for a problem and $\mathcal{A}$ be the set of all deterministic algorithms solving it. Then,
for any $x \in \mathcal{X}$ and $A \in \mathcal{A}$, the cost of running $A$ on $x$ is denoted by $c(a,x) \ge 0$.

Let $p$ be the probability distribution over $\mathcal{A}$, and let $A$ be an algorithm chosen at random according to $p$. Let $q$ be
the probability distribution over $\mathcal{X}$, and let $X$ be an input chosen at random according to $q$. Then
the worst-case expected cost of the randomized algorithm is at least as large as the cost of the best deterministic algorithm against the chosen
distribution on the inputs:
\[\max_{x \in \mathcal{X}} \mathbf{E}[c(A,x)] \ge \min_{a \in \mathcal{A}} \mathbf{E}[c(a,X)].\]
\end{theorem}

We use the above theorem for both Las Vegas and Monte Carlo algorithms. For Las Vegas algorithms, we consider only correct algorithms, and $c(x,a)$ is the
memory usage.
For Monte Carlo algorithms, we consider all algorithms (not necessarily correct) with memory usage not exceeding a certain threshold,
and $c(x,a)$ is the correctness indicator function, \ie, $c(x,a)=0$ if the algorithm is correct and $c(x,a)=1$ otherwise.

Our proofs will be based on appropriately chosen padding. The padding requires a constant number
of fresh characters.
If $\Sigma$ is twice as large as the number of required fresh characters, we can still use half of
it to construct a difficult input instance, which does not affect the asymptotics. Otherwise,
we construct a difficult input instance over $\Sigma$, then add enough new fresh characters
to facilitate the padding, and finally reduce the resulting larger alphabet to binary at the expense
of increasing the size of the input by a constant factor.

\begin{lemma}
\label{alphabetreduction}
For any alphabet $\Sigma=\{1,2,\ldots,\sigma\}$ there exists a morphism $h : \Sigma^* \rightarrow \{0,1\}^*$ such that, for any $c\in\Sigma$, $|h(c)| = 2\sigma+6$ and, for any word $w$, $w$ contains a palindrome
of length $\ell$ if and only if $h(w)$ contains a palindrome of length $(2\sigma+6)\cdot \ell$.
\end{lemma}

\begin{proof}
We set:
\[h(c) = 1 1^s 0 1^{s-c} 1 00 1 1^{s-c} 0 1^c 1.\]
Clearly $|h(c)|=2\sigma+6$ and, because every $h(c)$ is a palindrome, if $w$ contains a palindrome
of length $\ell$ then $h(w)$ contains a palindrome of length $(2\sigma+6)\cdot\ell$. Now assume that
$h(w)$ contains a palindrome of length $(2\sigma+6)\cdot\ell$, where $\ell \geq 1$. 
If $\ell=1$ then we obtain that $w$ should contain a palindrome of length $1$, which always holds.
Otherwise, the palindrome contains $00$ inside and we consider two cases.
\begin{enumerate}
\item The palindrome is centered inside $00$. Then it corresponds to an odd palindrome of length
$\ell$ in $w$.
\item The palindrome maps some $00$ to another $00$. Then it corresponds to an even palindrome
of length $\ell$ in $w$.
\end{enumerate}
In either case, the claim holds.
\end{proof}

For the padding we will often use an infinite word $\nu = 0^11^10^21^20^31^3\ldots$, or more precisely
its prefixes of length $d$, denoted $\nu(d)$. Here $0$ and $1$ should be understood as two characters
not belonging to the original alphabet, which is then reduced using the above lemma. The longest
palindrome inside $\nu(d)$ has radius $\bigo(\sqrt{d})$. 
% Pawel: exact constant in the bigo?

We first show that any Las Vegas approximation algorithm must necessarily use
$\Omega(n\log|\Sigma|)$ bits of memory in expectation in both variants, so Las Vegas
randomization is essentially useless here.
By Yao's minimax principle, it is enough to construct a distribution over the inputs, which is hard for
any deterministic algorithm using less memory. We restrict the inputs to a family of strings
of the form $\nu(\aerr) x \$\$ x^R \nu(\aerr)^R$, where $\$$ is a special character not belonging to
$\Sigma$ and $\nu(\aerr)$ is a padding word of length $\aerr$ chosen so that there are no long palindromes
inside. Then the longest palindrome must be centered in the middle of the whole word.
By a counting argument, the state of the algorithm after
having seen $\nu(\aerr) x \$$ must be distinct for different words $x$, so the required number
of bits is $\Omega(n\log |\Sigma|)$ in expectation.
A bound on multiplicative approximation follows because multiplicative approximation implies
additive approximation.

\begin{restatable}[Las Vegas approximation]{theorem}{restatablelv}
\label{th:lv}
Let $\alg$ be a Las Vegas streaming algorithms solving \palin
with additive error $\aerr \le 0.49 n$ or multiplicative error $(1+\varepsilon) \le 50$
using $s(n)$ bits of memory.
Then $\mathbb{E}[s(n)]=\Omega(n \log |\Sigma|)$.
\end{restatable}

\begin{proof}
By Theorem~\ref{yao}, it is enough to construct a probability distribution $\mathcal{P}$ over $\Sigma^n$ such that
for any deterministic algorithm $\dalg$, its expected memory usage on a word chosen according to $\mathcal{P}$ is
$\Omega(n \log |\Sigma|)$ in bits.

Consider solving \palin with additive error $\aerr$.
We define $\mathcal{P}$ as the uniform distribution over $\nu(\aerr) x \$\$ y \nu(\aerr)^R$, where $x,y \in \Sigma^{n'}$, $n' = \frac{n}{2}-\aerr-1$, and $\$$ are special characters not belonging to $\Sigma$.
Let us look at the memory usage of $\dalg$ after having read $\nu(\aerr) x$. We say that $x$ is "good" when the memory usage 
is at most $\frac{n'}{2}\log |\Sigma|$ and "bad" otherwise.
Assume that $\frac{1}{2}|\Sigma|^{n'}$ of all $x$'s are good, then there are two strings
$x \not= x'$ such that the state of $\dalg$ after having read both $\nu(\aerr) x$ and $\nu(\aerr) x'$ is exactly the same. Hence the behavior of $\dalg$ on
$\nu(\aerr) x\$\$ x^R \nu(\aerr)^R$ and $\nu(\aerr) x'\$\$ x^R \nu(\aerr)^R$ is exactly the same. The former is a palindrome of radius $\frac{n}{2} = n'+\aerr+1$, so $\dalg$ 
must answer at least  $n'+1$, and consequently the latter also must contain a palindrome of radius at least $n'+1$.
A palindrome inside $\nu(\aerr) x'\$\$ x^R \nu(\aerr)^R$ is either fully contained within
$\nu(\aerr)$, $x'$, $x^R$ or it is a middle palindrome.
But the longest palindrome inside $\nu(\aerr)$ is of length $\bigo(\sqrt{\aerr})<n'+1$ (for $n$ large enough)
and the longest palindrome inside $x$ or $x^R$ is of length $n'<n'+1$, so
$\nu(\aerr) x'\$\$ x^R \nu(\aerr)^R$
contains a middle palindrome of radius $n'+1$. This implies that $x=x'$, which is a contradiction.
Therefore, at least $\frac{1}{2}|\Sigma|^{n'}$ of all $x$'s are bad. But then the expected memory usage of $\dalg$ is at least
$\frac{n'}{4}\log |\Sigma|$, which for $\aerr\le0.49 n$ is $\Omega(n \log |\Sigma|)$ as claimed.

Now consider solving \palin with multiplicative error $(1+\varepsilon)$. 
An algorithm with multiplicative error $(1+\varepsilon)$ can also be considered as having additive error $\aerr=\frac{n}{2} \cdot \frac{\varepsilon}{1+\varepsilon}$, so if the expected memory usage of
such an algorithm is $o(n\log |\Sigma|)$ and $(1+\varepsilon) \le 50$ then we obtain
an algorithm with additive error $\aerr \le \frac{n}{2}\frac{49}{50}=0.49n$ and expected memory usage $o(n\log |\Sigma|)$, which we already know to be impossible.
\end{proof}

Now we move to Monte Carlo algorithms. We first consider exact algorithms solving \palin;
lower bounds on approximation algorithms will be then obtained by padding the input
appropriately. We introduce an auxiliary problem \midpalin, which is to compute radius
of the middle palindrome in a word of length $n$ over an alphabet $\Sigma$. We want to
show that solving \midpalin exactly with error probability smaller than $\frac{1}{n|\Sigma|}$ requires
$\lfloor\frac{n}{2} \log |\Sigma| \rfloor$ bits of space. By Yao's minimax principle,
it is enough to construct a distribution over the inputs, such that any deterministic algorithm
using less memory is not able to distinguish between inputs with different answers
reasonably often. This can be done by considering uniform distribution on inputs of the form
$x[1]\ldots x[\frac{n}{2}]x[\frac{n}{2}]\ldots x[k+1]c x[k-1] \ldots x[1]$. Then amplification
(running multiple instances of an algorithm in parallel) gives us a lower bound on the space
complexity of any algorithm solving \midpalin exactly. The lower bound can be translated
to \palin by padding the input in the middle, so that the longest palindrome must be centered
in the middle.

\begin{restatable}{lemma}{restatablelowerboundexact}
\label{lowerbound:exact}
There exists a constant  $\gamma$ such that any randomized Monte Carlo streaming algorithm
$\alg$ solving \midpalin or \palin exactly with probability $1-\frac{1}{n}$
uses at least $\gamma \cdot n \log\min \{|\Sigma|, n\}$ bits of memory.
\end{restatable}

\begin{proof}
First we prove that if $\alg$ is a Monte Carlo streaming algorithm solving \midpalin exactly using less than $\lfloor\frac{n}{2} \log |\Sigma| \rfloor$ bits
of memory, then its error probability is at least $\frac{1}{n|\Sigma|}$.

By Theorem~\ref{yao}, it is enough to construct probability distribution $\mathcal{P}$ over $\Sigma^n$ such that for any deterministic
algorithm $\dalg$ using less than $\lfloor\frac{n}{2} \log |\Sigma| \rfloor$ bits of memory, the expected probability of error on a word chosen
according to $\mathcal{P}$ is at least $\frac{1}{n|\Sigma|}$.

Let $n' = \frac{n}{2}$. For any $x\in\Sigma^{n'}$, $k\in\{1,2,\ldots,n'\}$ and $c\in\Sigma$ we define:
\[w(x,k,c) = x[1] x[2] x[3]  \ldots x[n'] x[n'] x[n'-1] x[n'-2] \ldots  x[k+1]  c  x[k-1] \ldots x[2]x[1].\]
Now $\mathcal{P}$ is the uniform distribution over all such $w(x,k,c)$.

Since there are $|\Sigma|^{n'}=2^{n' \log |\Sigma|} \ge 2\cdot 2^{\lfloor\frac{n}{2} \log |\Sigma| \rfloor-1}$ possible strings of length $n'$ and we assume that
$\dalg$ uses at most $\lfloor\frac{n}{2} \log |\Sigma| \rfloor$ bits,
we can 
partition at least half of these strings into pairs $(x,x')$, such that $\dalg$ is in the same state after reading either $x$ or $x'$. (If we choose an arbitrary maximal
matching of strings into pairs, at most half of possible strings will be left unpaired, that is one per each possible state of $\dalg$.)
Let $s$ be longest common suffix of $x$ and $x'$, so $x = v c s$ and $x' = v' c' s$, where $c \not= c'$ are single characters.  Then
$\dalg$ returns the same answer on $w(x,n'-|s|,c)$ and $w(x',n'-|s|,c)$, even though the radius of the middle palindrome is exactly $|s|$ in one of
them, and at least $|s|+1$ in the other one. Therefore, $\dalg$ errs on at least one of these two inputs. Similarly, it errs on either
$w(x,n'-|s|,c')$ or $w(x,n'-|s|,c')$. Thus the error probability is at least $\frac{1}{2n'|\Sigma|} = \frac1{n|\Sigma|}$. 

Now we can prove the lemma for \midpalin with a standard amplification trick.
Say that we have a~Monte Carlo streaming algorithm, which solves \midpalin exactly
with error probability $\varepsilon$ using $s(n)$ bits of memory. Then we can run its $k$ instances simultaneously and return
the most frequently reported answer. The new algorithm needs $\bigo(k\cdot s(n))$ bits of memory and its error probability $\varepsilon_{k}$ satisfies:
\[\varepsilon_k \le \sum_{2i < k} \binom{k}{i}(1-\varepsilon)^i \varepsilon^{k-i} \le 2^k \cdot \varepsilon^{k/2} = (4 \varepsilon)^{k/2}.\]

Let us choose $\kappa = \frac{1}{6}\frac{\log(4/n)}{\log(1/(n|\Sigma|))} = \frac16 \frac{1-o(1)}{1+\log|\Sigma|/\log n} = \Theta(\frac{\log n}{\log n + \log |\Sigma|}) = \gamma \cdot \frac{1}{\log |\Sigma|} \log \min \{|\Sigma|, n\}$, for some constant $\gamma$.
Now we can prove the theorem. Assume that $\alg$ uses less than $\kappa \cdot n\log |\Sigma| = \gamma \cdot n \log \min  \{|\Sigma|, n\}$ bits of memory. Then running $\left\lfloor \frac{1}{2\kappa} \right\rfloor \ge \frac34\frac1{2\kappa}$ (which holds since $\kappa < \frac16$)
instances of $\alg$ in parallel requires less than $\lfloor \frac{n}{2} \log |\Sigma| \rfloor$ bits of memory. But then 
the error probability of the new algorithm is bounded from above by:
\[\left(\frac{4}{n}\right)^{\frac{3}{16\kappa}} = \left(\frac{1}{n|\Sigma|}\right)^{\frac{18}{16}} \le \frac{1}{n|\Sigma|}\]
which we have already shown to be impossible.

The lower bound for \midpalin can be translated into a lower bound for solving \palin exactly
by padding the input so that the longest palindrome is centered in the middle.
Let $n'=\frac{n}{2}$ and
$x=x[1]x[2]\ldots x[n]$ be the input for \midpalin. We define:
\[w(x)= x[1] x[2] x[3]  \ldots x[n']  1 \underbracket{000 \ldots 0}_{n} 1 x[n'+1] \ldots x[n].\]
Now if the radius of the middle palindrome in $x$ is $k$, then $w(x)$ contains a palindrome of radius
at least $n'+k+1$. In the other direction, any palindrome inside $w(x)$ of radius larger than $n'$ must be centered somewhere
in the middle block consisting of only zeroes and both ones are mapped to each other, so it must be
the middle palindrome.
Thus, the radius of the longest palindrome inside $w(x)$ is exactly $n'+k+1$, so we have reduced solving \midpalin to solving
\palin[2n+2]. We already know that solving \midpalin[n] with probability $1-\frac{1}{n}$
requires $\gamma \cdot n \log \min \{|\Sigma|, n\}$ bits of memory,
so solving \palin[2n+2] with probability $1-\frac{1}{2n+2}\geq 1-\frac{1}{n}$ requires $\gamma \cdot n \log \{|\Sigma|,n\} \geq \gamma' \cdot (2n+2) \log \min \{|\Sigma|, 2n+2\}$ bits of memory.
Notice that the reduction needs $\bigo(\log n)$ additional bits of memory to count up to $n$, but for large $n$ this is
much smaller than the lower bound if we choose $\gamma' < \frac{\gamma}{4}$.
\end{proof}

To obtain a lower bound for Monte Carlo additive approximation, we observe that any algorithm
solving \palin with additive error $\aerr$ can be used to solve \palin[\frac{n-\aerr}{\aerr+1}] exactly
by inserting $\aerr$ zeroes between every two characters, in the very beginning, and in the very end.
However, this reduction requires $\log \aerr\leq \log n$ additional bits of memory for counting up to $\aerr$
and cannot be used when the desired lower bound on the required number of bits
$\Omega(\frac{n}{\aerr}\log\min(|\Sigma|,\frac{n}{E})$ is significantly smaller than $\log n$. 
Therefore, we need a separate technical lemma which implies that either additive or multiplicative
approximation with error probability $\frac{1}{n}$ requires $\Omega(\log n)$ bits of space.

\begin{restatable}{lemma}{restatablehashinglb}
\label{hashing_lb}
Let $\alg$ be any randomized Monte Carlo streaming algorithm solving \palin with additive error at most $0.49 n$ or multiplicative error at most
$n^{0.49}$  and error probability $\frac{1}{n}$.
Then $\alg$ uses $\Omega(\log n)$ bits of memory.
\end{restatable}

\begin{proof}
By Theorem~\ref{yao}, it is enough to construct a probability distribution $\mathcal{P}$ over $\Sigma^n$, such that for
any deterministic algorithm $\dalg$ using at most $s(n)=\bigo(\log n)$ bits of memory, the expected probability of error on a word chosen according
to $\mathcal{P}$ is $\frac{1}{2^{s(n)+2}}$.

Let $n' = s(n)+1$.  For any $x,y \in \Sigma^{n'}$, let $w(x,y) = \nu(\frac{n}{2}-n')^R x y^R \nu(\frac{n}{2}-n')^R$.
Observe that if $x=y$ then $w(x,y)$ contains a palindrome of radius $\frac{n}{2}$, and otherwise the longest palindrome there has
radius at most $2n'+\bigo(\sqrt{n}) = \bigo(\sqrt{n})$, thus any algorithm with additive error of at most $0.49  n$ or with a multiplicative error at most $n^{0.49}$
must be able to distinguish between these two cases (for $n$ large enough).

Let $S \subseteq \Sigma^{n'}$ be an arbitrary family of words of length $n'$ such that $|S| = 2 \cdot 2^{s(n)}$, and let $\mathcal{P}$ be 
the uniform distribution on all words of the form $w(x,y)$, where $x$ and $y$ are chosen uniformly and independently from $S$.
By a counting argument, we can create at least $\frac{|S|}{4}$ pairs $(x,x')$ of elements
from $S$ such that the state of $\dalg$ is the same after having read $\nu(\frac{n}{2}-n')^Rx$ and $\nu(\frac{n}{2}-n')^Rx'$. 
(If we create the pairs greedily, at most one such $x$ per state of memory can be left unpaired, so at least $|S| - 2^{s(n)} = \frac{|S|}2$ elements are paired.)
Thus, $\dalg$ cannot
distinguish between $w(x,x')$ and $w(x,x)$, and between $w(x',x')$ and $w(x',x)$, so its error probability must be at least
$\frac{|S|/2}{|S|^2} = \frac{1}{4\cdot 2^{s(n)}}$. Thus if $s(n) = o(\log n)$, the error rate is at least $\frac{1}{n}$ for $n$ large enough, a contradicion.
\end{proof}

Combining the reduction with the technical lemma and 
taking into account that we are reducing to a problem with word length of $\Theta(\frac{n}{E})$,
we obtain the following.

\begin{restatable}[Monte Carlo additive approximation]{theorem}{restatablemontecarloadditive}
\label{th:montecarlo_additive_lowerbound}
Let $\alg$ be any randomized Monte Carlo streaming algorithm solving \palin with additive error $\aerr$
with probability $1-\frac{1}{n}$. If $\aerr \le 0.49n$ then $\alg$ uses
$\Omega(\frac{n}{\aerr} \log \min \{|\Sigma|,\frac{n}{\aerr}\} )$ bits of memory.
\end{restatable}

\begin{proof}
%First, observe that the $\Omega(\log n)$ part follows trivially from Lemma~\ref{hashing_lb}, as any algorithm with error probability
%bounded by $\frac{1}{n}$ have to use $\Omega(\log n)$ bits of memory.
Define $\sigma = \min \{|\Sigma|,\frac12 \frac{n}{\aerr}\}.$

Because of Lemma~\ref{hashing_lb} and $\log\sigma \geq \frac{1}{35}\log \min\{|\Sigma|,\frac{n}{\aerr}\}$
(which holds due to $\aerr \leq 0.49n$),
it is enough to prove that $\Omega(\frac{n}{\aerr} \log \sigma)$ is a lower bound when 
\begin{equation}\label{foo} \aerr \le \frac{\gamma}4 \cdot \frac{n}{\log n} \log \sigma. \end{equation} 
%To prove it we will pad the input word with separating zeroes, which will allow us to use the lowerbound for the exact version.
Assume that there is a Monte Carlo streaming algorithm $\alg$ solving \palin with additive error $\aerr$ using $o(\frac{n}{\aerr}\log \sigma)$ 
bits of memory and probability $1-\frac{1}{n}$. 
Let $n' = \frac{n-\aerr}{\aerr+1} \ge \frac12 \frac{n}{\aerr}$ (the last inequality holds because $\aerr \le 0.49n$ and because we can assume that $\aerr > 1$).  Given a word $x[1] x[2] \ldots x[n']$, we can simulate running $\alg$
on $0^{\aerr} x[1] 0^{\aerr} x[2] 0^{\aerr} x[3] \ldots 0^{\aerr} x[n'] 0^{\aerr}$ to calculate $R$ (using $\log \aerr \leq \log n$ additional bits of memory), and then return $\left\lfloor \frac{R}{\aerr+1}\right\rfloor$.
We call this new Monte Carlo streaming algorithm $\alg'$.
Recall that $\alg$ reports the radius of the longest palindrome with additive error $\aerr$. Therefore, if the original word contains
a palindrome of radius $r$, the new word contains a palindrome of radius $\frac{\aerr}{2}+r(\aerr+1)$, so $R\geq r(\aerr+1)$ and $\alg'$
will return at least $r$. In the other direction, if $\alg'$ returns $r$, then the new word contains a palindrome of radius $r(\aerr+1)$.
If such palindrome is centered so that $x[i]$ is matched with $x[i+1]$ for some $i$, then
it clearly corresponds to a palindrome of radius $r$ in the original word. But otherwise every
$x[i]$ within the palindrome is matched with $0$, so in fact the whole palindrome corresponds to a streak
of consecutive zeroes in the new word and can be extended to the left and to the right to start and end
with $0^{\aerr}$, so again it corresponds to a palindrome of radius $r$ in the original word.
Therefore, $\alg'$ solves \palin[n'] exactly with probability
$1-\frac{1}{(n'(\aerr+1)+\aerr)} \ge 1 - \frac{1}{n'}$ and uses
$o(\frac{n'(\aerr+1)+\aerr}{\aerr} \log \sigma)+\log n = o(n' \log \sigma) + \log n$ bits of memory. Observe that by Lemma~\ref{lowerbound:exact} we get a lower bound 
\[
\gamma \cdot n' \log \min \{|\Sigma|,n'\} \geq \frac{\gamma}{2}\cdot n' \log \sigma+ \frac{\gamma}{4}\cdot \frac{n}{E} \log\sigma
\ge \frac{\gamma}{2}\cdot n' \log \sigma + \log n
\label{eq:lb}
\]
 (where the last inequality holds because of Eq.\eqref{foo}). Then, for large $n$
we obtain contradiction as follows
\[
o(n' \log\sigma) + \log n < \frac{\gamma}{2} \cdot n' \log\sigma + \log n. \qedhere
\]
\end{proof}

Finally, we consider multiplicative approximation. Here we observe that any algorithm solving
\palin with multiplicative error $(1+\varepsilon)$ can be used  to solve \midpalin[2n'] exactly, where
$n'=\Theta(\frac{\log n}{\log(1+2\varepsilon)})$, by separating the characters appropriately.
Intuitively, the padding is chosen so that the middle palindrome has the largest radius and
the larger the distance from the center the longer the separator inserted between two consecutive
characters of the original input. Again, we need $\log n$ bits for a counter and hence need
to invoke a separate technical lemma when $(1+\varepsilon)$ is very large. After some calculations,
and taking into account that we are reducing to a problem with word length of $\Theta(\frac{\log n}{\log(1+\varepsilon)})$
we obtain the following.

\begin{restatable}[Monte Carlo multiplicative approximation]{theorem}{restatablemontecarlomultiplicative}
\label{th:montecarlo_multiplicative_lowerbound}
Let $\alg$ be any randomized Monte Carlo streaming algorithm solving \palin with multiplicative error $(1+\varepsilon)$
with probability $1-\frac{1}{n}$. If $n^{-0.98} \le \varepsilon \le n^{0.49}$ then $\alg$ uses $\Omega(\frac{\log n}{\log(1+\varepsilon)}\log \min \{|\Sigma|,\frac{\log n}{\log(1+\varepsilon)}\})$ bits of memory.
\end{restatable}

\begin{proof}
For $\varepsilon \ge n^{0.001}$ then the claimed lower bound reduces to $\Omega(1)$ bits,
which obviously holds. Thus we can assume that $\varepsilon < n^{0.001}$.
Define \[\sigma = \min \{|\Sigma|, \frac{1}{50} \frac{\log n}{\log(1+2\varepsilon)}-2\}.\]
First we argue that it is enough to prove that $\mathcal{A}$ uses
$\Omega(\frac{\log n}{\log(1+\varepsilon)}\log\sigma)$ bits of memory.
Since $\log (1+2\varepsilon) \le 0.001 \log n + o(\log n)$, we have that:
\begin{equation}
\label{eq_18}
\frac{1}{50}\frac{\log n}{\log(1+2\varepsilon)}-2  \ge 18 - o(1)
\end{equation}
and consequently:
\begin{equation}
\label{eq:sigma1}
\frac{1}{50}\frac{\log n}{\log(1+2\varepsilon)}-2 = \Theta(\frac{\log n}{\log(1+2\varepsilon)}).
\end{equation}
Finally, observe that:
 \begin{equation}
 \label{eq:2eps}
 \log(1+2\varepsilon) = \Theta(\log(1+\varepsilon))
 \end{equation}
because $\log 2(1+\varepsilon) = \Theta(\log(1+\varepsilon))$ for $\varepsilon \ge 1$,
and $\log(1+\varepsilon) = \Theta(\varepsilon)$ for $\varepsilon < 1$.
From \eqref{eq:sigma1} and \eqref{eq:2eps} we conclude that:
\begin{equation}
\label{eq:sigma2}
 \log\sigma = \Theta(\log \min\{|\Sigma|,  \frac{\log n}{\log(1+\varepsilon)} \}).
 \end{equation}

Because of Lemma~\ref{hashing_lb} and equations \eqref{eq:2eps} and \eqref{eq:sigma2},
it is enough to prove that $\Omega(\frac{\log n}{\log(1+\varepsilon)}\log\sigma)$ is a lower bound when
\begin{equation}
\label{foo2}
\log(1+2\varepsilon) \le \gamma \cdot \frac{\log\sigma}{100},
\end{equation} 
as otherwise $\Omega(\frac{\log n}{\log (1+\varepsilon)} \log \sigma ) = \Omega( \frac{\log n}{\log (1+2\varepsilon)} \log \sigma) = \Omega(\log n)$.

Assume that there is a Monte Carlo streaming algorithm $\alg$ solving \palin with multiplicative error $(1+\varepsilon)$
with probability $1-\frac{1}{n}$ using $o(\frac{\log n}{\log(1+\varepsilon)}\log \sigma )$ bits of memory.
Let $x = x[1] x[2] \ldots x[n'] x[n'+1] \ldots x[2n']$ be an input for \midpalin[2n']. We choose $n'$ so that $n=(1+2\varepsilon)^{n'+1} \cdot n^{0.99} $. Then $n'= \log_{(1+2\varepsilon)} (n^{0.01})-1 = \frac{1}{100} \frac{\log n}{\log(1+2\varepsilon)}-1$.
We choose $i_0,i_1, i_{2},i_{3},\ldots,i_{n'}$ so that  $i_0+\ldots+i_d = \lceil(1+2\varepsilon)^{d+1} \cdot n^{0.99}\rceil$ for any $0 \le d \le n'$.

%Then, for any $d \ge 0$, $i_0+\ldots+i_d = (1+2\varepsilon)^{d+1} \cdot n^{0.99}$.
(Observe that for $\varepsilon = \Omega(n^{-0.98})$ we have $i_0 > n^{0.99}$ and $i_1,\ldots,i_d > 2n^{0.01}-1$.) Finally we define:
\[w(x) = \nu(i_{n'})^R x[1] \nu(i_{n'-1})^R  \ldots x[n'] \nu(i_{0})^R \nu(i_{0}) x[n'+1] \nu(i_{1}) \ldots \nu(i_{n'-1}) x[2n'] \nu(i_{n'}).\]

If $x$ contains a middle palindrome of radius exactly $k$, then $w(x)$ contains a middle palindrome of radius  $(1+2\varepsilon)^{k+1} \cdot n^{0.99}$.
Also, based on the properties of $\nu$, any non-middle centered palindrome in $w(x)$ has radius at most $\bigo(\sqrt{n})$, which is less than $n^{0.99}$ for $n$ large 
enough. Since $\lceil(1+2\varepsilon)^{k}\cdot n^{0.99}\rceil \cdot (1+\varepsilon) < ((1+2\varepsilon)^{k}\cdot n^{0.99}+1) \cdot (1+\varepsilon) <(1+2\varepsilon)^{k+1}\cdot n^{0.99}$, value of $k$ can be extracted from the answer of $\alg$.
Thus, if $\alg$ approximates the middle palindrome in $w(x)$ with multiplicative error $(1+\varepsilon)$ with probability $1-\frac{1}{n}$
using
$o(\frac{\log n}{\log(1+\varepsilon)}\log\sigma)$ bits of memory, we can construct a new algorithm $\alg'$ solving \midpalin[2n'] exactly with
probability $1-\frac{1}{n}>1-\frac{1}{2n'}$ using
\begin{equation}
\label{eq:ub2}
o(\frac{\log n}{\log(1+\varepsilon)}\log\sigma) + \log n
\end{equation}
bits of memory.
By Lemma~\ref{lowerbound:exact}  we get a lower bound
\begin{eqnarray}
\gamma \cdot 2n' \log \min \{|\Sigma|,2n'\}  &=& \frac{\gamma}{50} \cdot \frac{\log n}{\log(1+2\varepsilon)} \log\sigma - 2\gamma\log\sigma\nonumber \\
&\ge &\frac{\gamma}{100} \cdot \frac{\log n}{\log(1+2\varepsilon)} \log\sigma + \log n - 2\gamma\log\sigma
\label{eq:lb2}
\end{eqnarray}
(where the last inequality holds because of~(\ref{foo2})). On the other hand, for large $n$
\begin{eqnarray*}
&\frac{\gamma}{100}\cdot \frac{\log n}{\log(1+2\varepsilon)}\log\sigma  - 2\gamma\log\sigma + \log n=\left(\frac{1}{100}\frac{\log n}{\log(1+2\varepsilon)}-2\right)\gamma\log\sigma +\log n\\
&=\Theta\left(\frac{\log n}{\log(1+\varepsilon)}\log\sigma\right) + \log n
\end{eqnarray*}
so~\eqref{eq:lb2} exceeds~\eqref{eq:ub2}, a contradiction.
\end{proof}

\newpage

\section*{Acknowledgments}

The first author is currently holding a post-doctoral position at Warsaw Center of Mathematics and Computer Science.
However, most of this work has been done when the first author was at Max-Planck-Institut f\"{u}r Informatik and
the second author at LIF, CNRS – Aix Marseille University (supported by the Labex Archimède and by the ANR project MACARON (ANR-13-JS02-0002)).

The authors would like to thank Tomasz Syposz for a suggestion which allowed them to simplify the algorithm.

%\nocite{*}
\bibliography{main}

\begin{thebibliography}{10}

\bibitem{Apostolico}
Alberto Apostolico, Dany Breslauer, and Zvi Galil.
\newblock Parallel detection of all palindromes in a string.
\newblock {\em Theor. Comput. Sci.}, 141(1{\&}2):163--173, 1995.

\bibitem{Berenbrink}
Petra Berenbrink, Funda Erg{\"u}n, Frederik Mallmann-Trenn, and Erfan~Sadeqi
  Azer.
\newblock {Palindrome Recognition In The Streaming Model}.
\newblock In {\em STACS 2014}, volume~25 of {\em LIPIcs}, pages 149--161,
  Dagstuhl, Germany, 2014. Schloss Dagstuhl--Leibniz-Zentrum fuer Informatik.

\bibitem{Breslauer}
Dany Breslauer and Zvi Galil.
\newblock Real-time streaming string-matching.
\newblock {\em {ACM} Transactions on Algorithms}, 10(4):22, 2014.

\bibitem{DictionaryStream}
Rapha{\"{e}}l Clifford, Allyx Fontaine, Ely Porat, Benjamin Sach, and
  Tatiana~A. Starikovskaya.
\newblock Dictionary matching in a stream.
\newblock In Nikhil Bansal and Irene Finocchi, editors, {\em Algorithms - {ESA}
  2015 - 23rd Annual European Symposium, Patras, Greece, September 14-16, 2015,
  Proceedings}, volume 9294 of {\em Lecture Notes in Computer Science}, pages
  361--372. Springer, 2015.

\bibitem{KMismatch}
Rapha{\"{e}}l Clifford, Allyx Fontaine, Ely Porat, Benjamin Sach, and
  Tatiana~A. Starikovskaya.
\newblock The \emph{k}-mismatch problem revisited.
\newblock In Robert Krauthgamer, editor, {\em Proceedings of the Twenty-Seventh
  Annual {ACM-SIAM} Symposium on Discrete Algorithms, {SODA} 2016, Arlington,
  VA, USA, January 10-12, 2016}, pages 2039--2052. {SIAM}, 2016.

\bibitem{ErgunPeriodicity}
Funda Erg{\"{u}}n, Hossein Jowhari, and Mert Saglam.
\newblock Periodicity in streams.
\newblock In Maria~J. Serna, Ronen Shaltiel, Klaus Jansen, and Jos{\'{e}} D.~P.
  Rolim, editors, {\em Approximation, Randomization, and Combinatorial
  Optimization. Algorithms and Techniques, 13th International Workshop,
  {APPROX} 2010, and 14th International Workshop, {RANDOM} 2010, Barcelona,
  Spain, September 1-3, 2010. Proceedings}, volume 6302 of {\em Lecture Notes
  in Computer Science}, pages 545--559. Springer, 2010.

\bibitem{FiciGKK14}
Gabriele Fici, Travis Gagie, Juha K{\"{a}}rkk{\"{a}}inen, and Dominik Kempa.
\newblock A subquadratic algorithm for minimum palindromic factorization.
\newblock {\em J. Discrete Algorithms}, 28:41--48, 2014.

\bibitem{FineWilf}
N.~J. Fine and H.~S. Wilf.
\newblock Uniqueness theorems for periodic functions.
\newblock {\em Proceedings of the AMS}, 16:109--114, 1965.

\bibitem{GalilSeiferas}
Zvi Galil and Joel Seiferas.
\newblock A linear-time on-line recognition algorithm for ``palstar''.
\newblock {\em J. ACM}, 25(1):102--111, January 1978.

\bibitem{GawrychowskiSTACS}
Paweł Gawrychowski, Florin Manea, and Dirk Nowotka.
\newblock {Testing Generalised Freeness of Words}.
\newblock In {\em STACS 2014}, volume~25 of {\em LIPIcs}, pages 337--349,
  Dagstuhl, Germany, 2014. Schloss Dagstuhl--Leibniz-Zentrum fuer Informatik.

\bibitem{ISIBT14}
Tomohiro I, Shiho Sugimoto, Shunsuke Inenaga, Hideo Bannai, and Masayuki
  Takeda.
\newblock Computing palindromic factorizations and palindromic covers on-line.
\newblock In {\em CPM 2014}, volume 8486 of {\em Lecture Notes in Computer
  Science}, pages 150--161. Springer, 2014.

\bibitem{ParametrizedStreaming}
Markus Jalsenius, Benny Porat, and Benjamin Sach.
\newblock Parameterized matching in the streaming model.
\newblock In {\em 30th International Symposium on Theoretical Aspects of
  Computer Science, {STACS} 2013, February 27 - March 2, 2013, Kiel, Germany},
  pages 400--411, 2013.

\bibitem{KaplanSlowdown}
Haim Kaplan and Robert~E. Tarjan.
\newblock Persistent lists with catenation via recursive slow-down.
\newblock In {\em Proceedings of the Twenty-seventh Annual ACM Symposium on
  Theory of Computing}, STOC '95, pages 93--102, New York, NY, USA, 1995. ACM.

\bibitem{KR}
Richard~M. Karp and Michael~O. Rabin.
\newblock Efficient randomized pattern-matching algorithms.
\newblock {\em IBM Journal of Research and Development}, 31(2):249--260, 1987.

\bibitem{KMP}
Donald~E. Knuth, Jr. James H.~Morris, and Vaughan~R. Pratt.
\newblock Fast pattern matching in strings.
\newblock {\em SIAM Journal on Computing}, 6(2):323--350, 1977.

\bibitem{KosolobovRS15}
Dmitry Kosolobov, Mikhail Rubinchik, and Arseny~M. Shur.
\newblock $\text{Pal}^{k}$ is linear recognizable online.
\newblock In {\em {SOFSEM} 2015}, volume 8939 of {\em Lecture Notes in Computer
  Science}, pages 289--301. Springer, 2015.

\bibitem{Manacher}
Glenn~K. Manacher.
\newblock A new linear-time ``on-line'' algorithm for finding the smallest
  initial palindrome of a string.
\newblock {\em J. ACM}, 22(3):346--351, 1975.

\bibitem{PoratStreaming}
Benny Porat and Ely Porat.
\newblock Exact and approximate pattern matching in the streaming model.
\newblock In {\em 50th Annual {IEEE} Symposium on Foundations of Computer
  Science, {FOCS} 2009, October 25-27, 2009, Atlanta, Georgia, {USA}}, pages
  315--323. {IEEE} Computer Society, 2009.

\bibitem{Yao77}
Andrew Chi-Chih Yao.
\newblock Probabilistic computations: Toward a unified measure of complexity
  (extended abstract).
\newblock In {\em FOCS}, pages 222--227. IEEE Computer Society, 1977.

\end{thebibliography}

\end{document}